\documentclass[a4paper,11pt]{article}
\pdfoutput=1 

\usepackage{jheppub} 

\usepackage[T1]{fontenc} 
\usepackage{subcaption}
\usepackage{verbatim}
\usepackage{mathrsfs}
\usepackage{physics}
\usepackage{amsthm}
\usepackage{mathtools}
\usepackage{enumerate}

\renewcommand{\tilde}{\widetilde}
\renewcommand{\bar}{\overline}

\newcommand{\reals}{\mathbb{R}}
\newcommand{\comps}{\mathbb{C}}

\newcommand{\A}{\mathcal{A}}
\newcommand{\M}{\mathcal{M}}
\renewcommand{\H}{\mathcal{H}}

\newcommand{\im}{\operatorname{im}}

\newcommand{\B}{\mathcal{B}}

\newcommand{\R}{\mathcal{R}}

\renewcommand{\L}{\mathcal{L}}
\newcommand{\Z}{\mathcal{Z}}

\newcommand{\C}{\mathcal{C}}

\newtheoremstyle{breakthm}%
{}{}%
{\itshape}{}%
{\bfseries}{}
{\newline}{}

\theoremstyle{breakthm}
\newtheorem{theorem}{Theorem}

\newtheorem{lemma}[theorem]{Lemma}
\newtheorem{prop}[theorem]{Proposition}

\newtheoremstyle{break}%
{\topsep}{\topsep}%
{}{}%
{\bfseries}{}
{\newline}{}

\theoremstyle{break}
\newtheorem{definition}[theorem]{Definition}
\newtheorem{remark}[theorem]{Remark}
\newtheorem{fact}[theorem]{Fact}
\newtheorem{background}[theorem]{Background}

\numberwithin{theorem}{section}

\title{Notes on the type classification of von Neumann algebras}
\author{Jonathan Sorce}
\affiliation{Center for Theoretical Physics, Massachusetts Institute of Technology}
\abstract{These notes provide an explanation of the type classification of von Neumann algebras, which has made many appearances in recent work on entanglement in quantum field theory and quantum gravity.
The goal is to bridge a gap in the literature between resources that are too technical for the non-expert reader, and resources that seek to explain the broad intuition of the theory without giving precise definitions.
Reading these notes will provide you with: (i) an argument for why ``factors'' are the fundamental von Neumann algebras that one needs to study; (ii) an intuitive explanation of the type classification of factors in terms of renormalization schemes that turn unnormalizable positive operators into ``effective density matrices;'' (iii) a mathematical explanation of the different types of renormalization schemes in terms of the allowed traces on a factor; (iv) an intuitive characterization of type I and II factors in terms of their ``standard forms;'' and (v) a list of some interesting connections between type classification and modular theory, including the argument for why type III$_1$ factors are believed to be the relevant ones in quantum field theory.
None of the material is new, but the pedagogy is different from other sources I have read; it is most similar in spirit to the recent work on gravity and the crossed product by Chandrasekaran, Longo, Penington, and Witten.}

\makeatletter
\gdef\@fpheader{MIT preprint ID MIT-CTP/5527\vspace{3em}}
\makeatother
\begin{document}
\maketitle

\section{Introduction}

The goal of these notes is to unify, under a single roof, a number of pedagogical ideas about the type classification of von Neumann algebras that have been discussed in an as-needed manner in the physics literature.
In particular, I have in mind Witten's paper ``Gravity and the crossed product'' \cite{witten2022gravity}, which includes an illuminating comment on page two that the type of a von Neumann factor can be thought of in terms of whether ``[the algebra has] pure states, as well as... density matrices,'' ``[the] algebra does not have pure states, but it does have density matrices,'' or ``[the] algebra does not have pure states [or] density matrices.''
Throughout much of Witten's recent work both alone and with other authors \cite{witten2018aps, witten2022does, witten2022gravity, longo2022note, chandrasekaran2022algebra, chandrasekaran2022large, penington2023algebras}, and in the work by Leutheusser and Liu that inspired it \cite{leutheusser2021causal, leutheusser2021emergent}, other important properties of the type classification of von Neumann algebras have appeared; in particular, the properties of traces on a von Neumann factor of a given type.
These notions are certainly not new to physics, as the original papers on von Neumann algebra classification \cite{murray1936rings, murray1937rings, neumann1940rings, murray1943rings} were motivated in part by a desire to classify the algebraic structures that can appear in quantum mechanical systems.
Some of the original examples of von Neumann algebras were obtained from physical systems, e.g. \cite{araki1963representations}, and a great deal of work went into establishing that in many quantum field theories, the type of the von Neumann algebra associated to a spacetime region is a so-called hyperfinite type III$_1$ factor --- see for example \cite{araki1964type, driessler1977type, fredenhagen1985modular, buchholz1987universal, buchholz1995scaling, yngvason2005role}.

In these notes, I will try to explain the classification of von Neumann algebras in a physically motivated way that gives background to Witten's comment about the type of a factor being determined by the types of states it contains.
I will not motivate why a physicist would want to consider von Neumann algebras in the first place; this is done excellently in section 2.6 of \cite{witten2018aps}, and I can provide no justification that is superior to what is given in that text.
The plan of the paper is as follows.

\textit{Note: The reader who wishes to understand the theory as quickly as possible is encouraged to read sections 4, 5, 6.1, and 7.
Other sections can be thought of as supplementary.}
\begin{itemize}
	\item In \hyperref[sec:basic-facts]{section 2}, I list some basic facts about von Neumann algebras that will be needed in the rest of the text.
	\item In \hyperref[sec:factors]{section 3}, I explain the justification for focusing our attention on the ``factors'' that are the basic building blocks of all von Neumann algebras.
	\item In \hyperref[sec:big-picture]{section 4}, I introduce the real meat of the subject by describing a basic example of an algebra that contains no density matrices, but contains operators that can be thought of as ``renormalizable density matrices.''
	\item In \hyperref[sec:projector-types]{section 5}, I explain the type classification of factors due to Murray and von Neumann \cite{murray1936rings}, which involves an algebraic characterization of projection operators as having either ``effectively finite dimension'' or ``effectively infinite dimension,'' and explain how positive operators built up out of ``effectively finite'' projectors can be interpreted as renormalizable density matrices. 
	\item In \hyperref[sec:traces]{section 6}, the ``renormalizability'' heuristic is given mathematical teeth by introducing the unique renormalized trace functional on a von Neumann factor, and explaining how the positive operators of finite renormalized trace in a von Neumann factor can be interpreted as density matrices.
	(Note that the renormalized trace functional on a von Neumann algebra is generically different than the trace induced from an orthonormal basis on Hilbert space, which is why certain operators with naively infinite trace can be ``renormalized'' to finite trace operators.)
	\item To aid the reader in developing a good gut understanding of von Neumann factors, \hyperref[sec:standard-forms]{section 7} explains the standard forms of type I and type II$_{\infty}$ algebras, explaining how each of these algebras can be thought of as the algebra of bounded operators on a Hilbert space together with an ``internal algebra.''
	This internal algebra is trivial when the total algebra is of type I, and is of type II$_1$ when the total algebra is of type II$_{\infty}$.  
	\item Finally, in \hyperref[sec:miscellany]{section 8}, I collect additional miscellaneous ideas that may be of physical interest.
	I explain some connections between the type classification of von Neumann factors and the Tomita-Takesaki modular theory (\hyperref[sec:modular-theory]{section 8.1}), sketch the arguments for considering type III$_1$ factors in quantum field theory (\hyperref[sec:type-III1]{section 8.2}), and explain how type theory is modified for non-factor algebras (\hyperref[sec:non-factor]{section 8.3}).
\end{itemize}

Several appendices provide mathematical background.
\hyperref[app:math-background]{Appendix A} explains nets, operator topologies, and the algebraic form of the spectral theorem; it is needed for understanding some technical statements made in the text, and the dedicated reader may want to work through it before approaching the rest of the notes, but it can generally be consulted on an as-needed basis.
\hyperref[app:trace-math]{Appendix B} provides proofs of some properties of traces on von Neumann algebras that are used in section \ref{sec:traces}.

My understanding of this subject is the product of around four months spent reading mathematical texts.
The sources consulted in compiling these notes are \cite{rudin1974real, rudin1991functional, douglas1998banach, conway2000course, dixmier2011neumann, pedersen1979c, jones2003neumann, sunder2012invitation, takesaki2001theory, takesaki2003theory}.
While I found my time with these sources illuminating, many of them are quite dense, and difficult to approach from a ``physicist's perspective.''
It is my hope that these notes will provide the broad strokes of the theory of type classification in a way that is more technically precise than some discussions that have appeared so far in the physics literature, but sufficiently simple as to be pedagogically accessible.

I would like for these notes to be a useful reference for working physicists, so if you think anything is inaccurate or confusing, please send me a message and I will consider it in my revisions.

\subsection{Assumptions and notation}

The following is a list of symbols and concepts that are introduced in the text, and assumptions that are made throughout.

\begin{itemize}
	\item
	$\H$ denotes a Hilbert space, and $\B(\H)$ is the algebra of bounded operators on that space. (See definition \ref{def:bounded}.)
	
	\item 
	$\H$ is assumed to be separable, meaning that it has a finite or countable orthonormal basis.
	
	\item The identity operator on a Hilbert space is denoted by the numeral $1$.
	
	\item ``Operator'' almost always means ``bounded operator.''
	The exception is section \ref{sec:modular-theory}, where both unbounded and bounded operators appear, and I take care to always specify which kind of operator we are talking about.
	
	\item An overline over a set denotes the topological closure of that set.
	This appears for the first time in fact \ref{fact:polar-decomposition}.

	\item
	For a subset $M \subseteq \B(\H),$ $M'$ denotes the \textbf{commutant}, which is the set of operators in $\B(\H)$ that commute with every operator in $M$. (See definition \ref{def:commutant}.)

	\item
	A subalgebra $\A$ of $\B(\H)$ that is closed under adjoints and contains the identity is called a \textbf{von Neumann algebra} if $\A = \A''$; the symbol $\A$ is generally used to denote such an algebra. (See definition \ref{def:vN-algebra}.)
	
	\item The \textbf{center} of a von Neumann algebra is the set $\A \cap \A'$, and is denoted $\mathcal{Z}.$ (See definition \ref{def:center}.)
	
	\item Every time I write ``projector'' I mean ``orthogonal projector,'' i.e., a bounded operator $P$ with $P^2 = P = P^*.$\footnote{$P^*$ denotes the adjoint of $P$; see fact \ref{fact:adjoint}.}
	The projectors $\{P_{\alpha}\}$ are said to be ``pairwise orthogonal'' if they satisfy $P_{\alpha} P_{\beta} = \delta_{\alpha \beta} P_{\beta}.$
	
	\item A ``trace'' on the von Neumann algebra $\A$ is a $[0,\infty]$-valued function of the positive operators in $\A$,\footnote{An operator $T$ is said to be positive if it satisfies $\bra{\psi}T\ket{\psi} \geq 0$ for all $\ket{\psi} \in \H.$} satisfying appropriate linearity and cyclicity conditions (see definition \ref{def:trace}).
	A ``renormalized trace'' on $\A$ is a trace that satisfies additional conditions --- faithfulness, cleverness, and normality --- that make it a good tool for defining physical expectation values on $\A$ (see definition \ref{def:renormalized-trace}).
\end{itemize}

\section{Definitions and basic propositions concerning von Neumann algebras}
\label{sec:basic-facts}

I will now list basic definitions and facts that are important to understand for the remainder of these notes.
Some of these results are so fundamental that they can be found in the opening chapters of any text on functional analysis or operator theory.
For more obscure results I will sketch proofs or indicate where they can be found.
For intuition, the reader may wish to consult appendix A of \cite{harlow2017ryu}, which explains some of these facts in the special case of finite-dimensional von Neumann algebras.

It is not necessary to read this section in detail.
One may wish to skim through the statements before proceeding to section \ref{sec:factors}, and reference this section as needed throughout the rest of the notes.

\begin{definition} \label{def:bounded}
	Given a Hilbert space $\H$, a linear map $T : \H \to \H$ is \textbf{bounded} if there is a number $k$ such that for all unit vectors $\ket{\psi} \in \H$, we have
	\begin{equation}
		\lVert T \ket{\psi} \rVert \leq k.
	\end{equation}
	The infimum over all such $k$ is called the \textbf{operator norm} or just the \textbf{norm}, $\lVert T \rVert.$
	
	The space of all bounded operators is denoted $\B(\H).$
\end{definition}

\begin{fact}
	If $T$ and $S$ are bounded operators and $\alpha$ is a complex number, then we have
	\begin{align}
		\lVert \alpha T \rVert
			& = |\alpha| \lVert T \rVert, \\
		\lVert T + S \rVert
			& \leq \lVert T \rVert + \lVert S \rVert, \\
		\lVert T S \rVert
			& \leq \lVert T \rVert \lVert S \rVert.
		\end{align}
	Consequently, $\B(\H)$ is an algebra: if $T$ and $S$ are bounded operators and $\alpha$ is a complex number, then $\alpha T, T+S,$ and $TS$ are all bounded operators.
\end{fact}

\begin{fact} \label{fact:adjoint}
	For any bounded operator $T$, there exists a bounded operator $T^*$ called its \textbf{adjoint} satisfying
	\begin{equation}
		\braket{x}{T y} = \braket{T^* x}{y}.
	\end{equation}
	For $T,S$ bounded operators and $\alpha$ a complex number, we have $T^{**} =T,$ $(TS)^* = S^* T^*,$ $(T + S)^* = T^* + S^*,$ and $(\alpha S)^* = \bar{\alpha} S^*.$
	
	We also have $\lVert T^* \rVert = \lVert T \rVert.$
\end{fact}

\begin{definition} \label{def:*-subalgebra}
	A subset of $\B(\H)$ is said to be a \textbf{$*$-subalgebra} if it is closed under scalar multiplication, operator multiplication, operator addition, and adjoints, and contains the identity operator.
\end{definition}

\begin{definition} \label{def:commutant}
	Given a Hilbert space $\H$ and a subset $M \subseteq \B(\H)$, the \textbf{commutant} $M'$ is the set of all bounded operators that commute with everything in $M$, i.e.,
	\begin{equation}
		M' \equiv \{T \in \B(\H) \text{ such that } T S = S T \text{ for all } S \in M \}.
	\end{equation}
\end{definition}

\begin{definition} \label{def:vN-algebra}
	A $*$-subalgebra $\A \subseteq \B(\H)$ is a \textbf{von Neumann algebra} if it is equal to its own double commutant, i.e., $\A = \A''.$
\end{definition}

\begin{fact} \label{fact:spectral-theorem}
	Let $\A$ be a von Neumann algebra and $T \in \A$ be a normal operator (i.e., $TT^* = T^* T$, for example this is satisfied when $T = T^*$).
	For any bounded function $f : \comps \to \comps$, there exists a corresponding bounded operator $f(T)$ in $\A$.
	This operator is characterized by the following properties.
	\begin{itemize}
		\item For the identity function $\text{id} : \comps \to \comps,$ we have $\text{id}(T) = T.$
		\item For any bounded $f : \comps \to \comps,$ we have $\bar{f}(T) = f(T)^*,$ where $\bar{f}$ denotes the complex conjugate of $f$.
		For example, for $f(z) = i z,$ we have $\bar{f}(z) = - i \bar{z}.$
		\item For any bounded functions $f, g : \comps \to \comps,$ we have $(fg)(T) = f(T) g(T)$ and $(f+g)(T) = f(T) + g(T),$ as well as $(\alpha f)(T) = \alpha f(T)$ for any complex number $\alpha.$
		\item If $f, g : \comps \to \comps$ agree on the spectrum of $T$ up to a set of measure zero (see appendix \ref{app:spectral-theory} for explanation of these terms), then we have $f(T) = g(T).$
	\end{itemize}
\end{fact}
\begin{proof}[Sketch of proof]
	The proof of this statement is highly nontrivial, and in fact it is basically the entire content of the spectral theorem.
	While many analysis-oriented texts will frame the spectral theorem as a statement about representations of normal operators as multiplication operators on some function space $L^2(X)$, in a more algebraically focused reference like \cite{douglas1998banach}, the fact given above is taken as the core statement of the spectral theorem; the ``multiplication operator'' story is then derived as a consequence.
	
	Proofs of the spectral theorem in this form can be found in chapter 4 of \cite{douglas1998banach} and chapter 12 of \cite{rudin1991functional}, and further explanation of the algebraic framework for thinking about the spectral theorem is given in appendix \ref{app:spectral-theory}.
\end{proof}

\begin{remark}
	The ``multiplication operator'' version of the spectral theorem is the one that appeals most naturally to finite-dimensional intuition, because it has the flavor of ``diagonalizing an operator'' by finding a representation in which the operator acts by scalar multiplication.
	The reason for thinking of fact \ref{fact:spectral-theorem} as a version of the spectral theorem is that the main motivation for diagonalizing operators is so that we can apply functions to those operators by acting on eigenvalues.
	The algebraic version of the spectral theorem tells us how to apply those functions directly without needing to think about diagonalization.
	
	The translation between the ``algebraic'' and ``diagonalization'' versions of the spectral theorem works as follows.
	By the final bullet point in fact \ref{fact:spectral-theorem}, it suffices to consider bounded functions $f : \sigma(T) \to \comps,$ where $\sigma(T) \subseteq \comps$ denotes the spectrum of $T$.
	(See appendix \ref{app:spectral-theory} for definitions.)
	For any open subset $\omega \subseteq \sigma(T)$, we may consider the characteristic function $\chi_{\omega}$ that takes the value $1$ on $\omega$ and takes the value zero on $\sigma(T) - \omega.$
	It is easy to see from fact \ref{fact:spectral-theorem} that $\chi_{\omega}(T)$ is an orthogonal projector; the operator $\chi_{\omega}(T)$ is called a \textit{spectral projection} of $T.$
	It is also easy to see that the constant function $c_{1} : \sigma(T) \to \comps$ defined by $c_1(z) = 1$ is such that $c_1(T)$ is the identity operator on $\H.$
	If the spectrum of $T$ is discrete, then every point $\lambda \in \sigma(T)$ is an open set in $\sigma(T)$, and so there is a spectral projection $\chi_{\lambda}(T)$ satisfying $T \chi_{\lambda}(T) = \lambda \chi_{\lambda}(T).$
	We have $\sum_{\lambda} \chi_{\lambda} = c_1,$ hence $\sum_{\lambda} \chi_{\lambda}(T) = 1.$
	This is the usual form of the spectral theorem for operators with discrete spectra --- there exist projection operators $\chi_{\lambda}(T)$, which project onto subspaces where $T$ acts as multiplication by $\lambda,$ and such that the spectral subspaces $\chi_{\lambda}(T) \H$ collectively span the full Hilbert space $\H.$
	
	In the general case, where $T$ does not have a discrete spectrum, it does not necessarily have spectral projectors corresponding to individual eigenvalues.
	Instead, the operators $\chi_{\omega}(T)$ must project onto continuous subsets of $\sigma(T)$.
	In this case, one can still think of $\mu : \omega \mapsto \chi_{\omega}(T)$ as a \textit{projection-valued measure} on $\sigma(T)$, and write the formal integral $\int d\mu\, \chi_{\omega}(T) = 1$, which can be thought of as a completeness relation for the spectral projections of $T$.
	However, the interpretation of this formula is not as intuitive as in the discrete-spectrum case, making the algebraic version of the spectral theorem given in fact \ref{fact:spectral-theorem} superior for completely general bounded operators $T$.
	For more on projection-valued measures, see chapter 12 of \cite{rudin1991functional}.
\end{remark}

\begin{fact} \label{fact:unitaries}
	If $\A$ is a von Neumann algebra, then any $T \in \A$ can be written as a linear combination of two Hermitian operators in $\A$ or four unitary operators in $\A$. This is useful because one can usually prove statements about general operators in a von Neumann algebra by proving those statements for Hermitian or unitary operators and then showing that those statements are preserved under linear combinations.
\end{fact}
\begin{proof}
	The representation of $T$ as a linear combination of two Hermitian operators in $\A$ is provided by the standard decomposition
	\begin{equation}
		T = \frac{T+ T^*}{2} + i \frac{T- T^*}{2i}.
	\end{equation}

	In light of this, one can prove the claim about unitary operators by showing that every Hermitian $S \in \A$ can be written as a linear combination of two unitary operators in $\A$. This is furnished by the decomposition
	\begin{equation} \label{eq:T-as-unitary-sum}
		S = \lVert S \rVert \frac{U + U^*}{2}
	\end{equation}
	with
	\begin{equation}
		U = \frac{S + i \sqrt{\lVert S \rVert^2 - S^2}}{\lVert S \rVert}
	\end{equation}
	Some work must be done to show that this expression makes sense and that it indeed defines a unitary operator.
	In light of fact \ref{fact:spectral-theorem}, we see that $U$ can be interpreted as $\hat{f}(S)$ where $\hat{f} : \comps \to \comps$ is any bounded extension to $\comps$ of the function $f : [-\lVert S \rVert, \lVert S \rVert] \to \comps$ given by
	\begin{equation}
		f(x) = \frac{x + i \sqrt{\lVert S \rVert^2 - x^2}}{\lVert S \rVert}.
	\end{equation}
	So $\hat{f}(S)$ is in $\A$; its unitarity and the identity \eqref{eq:T-as-unitary-sum} follow from the properties listed in fact \ref{fact:spectral-theorem}.
	Note that the spectrum property listed in fact \ref{fact:spectral-theorem} means that the behavior of $\hat{f}$ away from $[-\lVert S \rVert, \lVert S \rVert]$ is irrelevant for defining $\hat{f}(S)$ thanks to the fact that $S$ is Hermitian and so its spectrum is contained in $[-\lVert S \rVert, \lVert S \rVert].$ (See appendix \ref{app:spectral-theory}.)
\end{proof}

\begin{fact} \label{fact:commutant-vN}
	If $\A$ is a von Neumann algebra, then so is $\A'.$
\end{fact}
\begin{proof}[Sketch of proof]
	This result is actually more general; if $\A$ is any subset of $\B(\H)$ closed under adjoints, then $\A'$ is a von Neumann algebra.
	The proof of this statement is nontrivial. 
	It uses the double commutant theorem, which says that any $*$-subalgebra of $\B(\H)$ that is topologically closed with respect to the ``weak operator topology''\footnote{The weak operator topology is explained in appendix \ref{app:operator-topologies}.} is a von Neumann algebra.
	One proves the fact quoted above by showing that the commutant of any adjoint-closed subset of $\B(\H)$ is a $*$-subalgebra that is closed in the weak operator topology, then appealing to the double commutant theorem.
	
	A proof of the double commutant theorem can be found in section 21 of \cite{conway2000course}.
	The claim about commutants of adjoint-closed sets being weak-operator-topology closed $*$-subalgebras is easy to prove from the definition of the weak operator topology, which is given in appendix \ref{app:operator-topologies}.
\end{proof}

\begin{fact}
	The intersection of any collection of von Neumann algebras is a von Neumann algebra.
\end{fact}
\begin{proof}[Sketch of proof]
	This can be shown by fairly straightforward manipulations starting with the definition of a von Neumann algebra.
	See proposition 1 in chapter I.1 of \cite{dixmier2011neumann}.
\end{proof}

\begin{definition} \label{def:center}
	Given a von Neumann algebra $\A$, the intersection $\Z = \A \cap \A'$ is a von Neumann algebra called the \textbf{center} of $\A$.
\end{definition}

\begin{definition}
	A von Neumann algebra $\A$ is a \textbf{factor} if $\Z$ consists only of scalar multiples of the identity.
\end{definition}

\begin{fact}
	$\A$ is a factor if and only if $\A$ and $\A'$ together generate all of $\B(\H),$ i.e., if the smallest von Neumann algebra containing both of them --- that is, $(\A \cup \A')''$ --- is equal to $\B(\H)$.
\end{fact}
\begin{proof}[Sketch of proof]
	The statement that $(\A \cup \A')''$ is the smallest von Neumann algebra containing both $\A$ and $\A'$ is an application of the more general fact that the smallest von Neumann algebra containing a set $M \subseteq \B(\H)$ is $(M \cup M^*)''$.
	This is straightforward to prove from the definitions, and an explicit proof can be found in the paragraph before proposition 1 in chapter I.1 of \cite{dixmier2011neumann}.
\end{proof}

\begin{fact} \label{fact:polar-decomposition}
	A bounded operator $T \in \B(\H)$ has a \textbf{polar decomposition}, which is an expression of the form $T = V |T|$ where $|T| = \sqrt{T^* T}$ and $V$ is a partial isometry\footnote{A partial isometry is an operator $V$ such that $V^* V$ is an orthogonal projection. Equivalently, it is an isometry when its domain is restricted to the orthogonal complement of its kernel.} with initial space $\ker(V)^{\perp} = \bar{\im(|T|)}$\footnote{An overline over a set denotes the topological closure of that set.} and final space $\im(V) = \bar{\im(T)}.$
	This decomposition is unique in the sense that if $T = U P$ is another decomposition into a partial isometry $U$ and a positive operator $P$, then we have $U = V$ and $P = |T|$ so long as $\ker(U)^{\perp}$ is equal to $\bar{\im(P)}.$
	
	If $\A$ is a von Neumann algebra and $T$ is in $\A$ with polar decomposition $T = V |T|,$ then both $|T|$ and $V$ are in $\A$.
\end{fact}
\begin{proof}
	The existence and uniqueness of the polar decomposition can be found in any text on operator theory.
	The fact that $|T|$ is in $\A$ follows from fact \ref{fact:spectral-theorem}.
	To show that $V$ is in $\A$, one shows that $V$ commutes with every unitary in $\A'$ and then appeals to fact \ref{fact:unitaries}.
	To show that $V$ commutes with every unitary in $\A'$, we let $U$ be one such unitary,and write
	\begin{equation}
		T = U^* T U = U^* V |T| U = U^* V U |T|.
	\end{equation}
	One can then show $U^* V U = V$ using uniqueness of the polar decomposition, for which it must be shown that the kernel of $U^* V U$ is the same as the kernel of $V$. $(1 - V^* V)$ is the projector onto the kernel of $V$, so if one shows that $U$ commutes with $V^* V$, then one has
	\begin{equation}
		U^* V U (1 - V^* V) = 0 \quad \text{and} \quad U^* V U (V^* V) = U^* V U,
	\end{equation}
	hence $\ker(V) \subseteq \ker(U^* V U)$ and $\ker(V)^{\perp} \subseteq \ker(U^* V U)^{\perp},$ so we may conclude $\ker(V) = \ker(U^* V U).$
	To show that $U$ commutes with $V^* V,$ we note that $V^* V$ is the projector onto the support of $T$, which is certainly in $\A$ as via fact \ref{fact:spectral-theorem} we have $V^* V = f(T)$ with
	\begin{equation}
		f(x) = \begin{cases}
						0 & x = 0,\\
						1 & \text{otherwise}.
				\end{cases}
	\end{equation}
	So $U$, being in $\A',$ commutes with $V^* V.$
\end{proof}

\begin{fact} \label{fact:invariant-commutant}
	If $\A$ is a von Neumann algebra and $P$ is an orthogonal projector in $\B(\H),$ then $P \H$ is an invariant subspace for $\A$ (i.e., $\A P \H \subseteq P \H$) if and only if we have $P \in \A'.$
\end{fact}
\begin{proof}
	If $P$ is in $\A'$ and $T$ is in $\A$, then we clearly have $T P \ket{x} = P T \ket{x}$ for all $\ket{x} \in \H$, hence $\A P \H \subseteq P \H.$
	
	Conversely, suppose we have $\A P \H \subseteq P \H$.
	We want to show that $P$ commutes with $\A$.
	By fact \ref{fact:unitaries}, it suffices to show that $P$ commutes with every Hermitian operator in $\A$.
	So let $T \in \A$ be Hermitian; the inclusion $\A P \H \subseteq P \H$ implies
	\begin{equation}
		P T P = T P.
	\end{equation}
	Taking adjoints gives
	\begin{equation}
		P T = P T P = T P,
	\end{equation}
	so $P$ commutes with $T$.
\end{proof}

\begin{fact} \label{fact:projector-infima-and-suprema}
	If $\{P_{\alpha}\}$ is a set of projectors in the von Neumann algebra $\A$, then the supremum\footnote{The supremum of a family of positive operators is the least positive operator that dominates it. An analogous definition holds for the infimum. Existence of these suprema and infima is part of the proof.} $\sup_{\alpha} \{P_{\alpha}\}$ and infimum $\inf_{\alpha} \{P_{\alpha}\}$ are projectors in $\A$.
\end{fact}
\begin{proof}
	Since each $P_{\alpha}$ is an orthogonal projector, its corresponding subspace $P_{\alpha} \H$ is a closed subspace of $\H$.
	The subspace $\bar{\sum_{\alpha} P_{\alpha} \H}$, which is the closure of all finite linear combinations of vectors in the various $P_{\alpha} \H$ subspaces, is a closed subspace.\footnote{This is important because of the basic fact in Hilbert space theory that a subspace of $\H$ has an orthogonal projector if and only if it is topologically closed.}
	We will call the orthogonal projector onto that subspace $\vee_{\alpha} P_{\alpha}.$
	The subspace $\cap_{\alpha} P_{\alpha} \H$ is also closed, so it has an orthogonal projector that we will call $\wedge_{\alpha} P_{\alpha}.$
	
	It is a straightforward exercise to show that $\vee_{\alpha} P_{\alpha}$ is the supremum of $\{P_{\alpha}\},$ and $\wedge_{\alpha} P_{\alpha}$ is its infimum.\footnote{For example, we clearly have $\wedge_{\alpha} P_{\alpha} \leq P_{\beta}$ for any $P_{\beta},$ and if $Q$ is a projector with $Q \leq P_{\alpha}$ for all $P_{\alpha},$ then we clearly have $Q \leq \wedge_{\alpha} P_{\alpha}.$ This establishes the ``infimum'' claim; the proof of the ``supremum'' claim is analogous.}
	What we need to show is that if all $\{P_{\alpha}\}$ projectors are in $\A,$ then $\vee_{\alpha} P_{\alpha}$ and $\wedge_{\alpha} P_{\alpha}$ are in $\A$ as well.
	By fact \ref{fact:invariant-commutant}, we can show this by showing that the subspaces $\vee_{\alpha} P_{\alpha} \H$ and $\wedge_{\alpha} P_{\alpha} \H$ are invariant subspaces for $\A'.$
	But this follows almost immediately from the observation (also from fact \ref{fact:invariant-commutant}) that $P_{\alpha} \in \A$ implies that each $P_{\alpha} \H$ is an invariant subspace for $\A'.$
	This last step is left as an exercise.
\end{proof}

\section{Why factors?}
\label{sec:factors}

In the remainder of these notes, excepting section \ref{sec:non-factor}, we will limit our attention to factors --- those von Neumann algebras for which $\Z  = \A \cap \A'$ consists only of scalar multiples of the identity.
There are two reasons for doing this, which I will explain in the next two subsections.
The first is that factors show up generically in quantum field theory; the second is that even when a von Neumann algebra is not a factor, it can be decomposed into factors.
Readers who are already happy restricting their attention to factors are encouraged to skip directly to section \ref{sec:big-picture}.

\subsection{Factors appear in quantum field theory}

Suppose we are studying quantum field theory in Minkowski spacetime, and we consider the von Neumann algebra $\A_\R$ of operators associated to a Rindler wedge $\R$.
Formally, we think of operators in $\A_R$ as being bounded functions of operators of the form
\begin{equation}
	\phi[f] = \int_{R} d^{d} x \sqrt{g} f(x) \phi(x),
\end{equation}
where $\phi(x)$ is a field and $f(x)$ is some function whose support is a compact set in $\R.$
See figure \ref{fig:smearing-1}.
For this discussion, it is very important that $\phi(x)$ itself is not actually an operator on Hilbert space, even an unbounded one, but rather an \textit{operator-valued distribution}.
This is because one of the fundamental properties of quantum field theory is that the leading singularity of the two-point function is universal among all states in the Hilbert space, that is, we have
\begin{equation}
	\bra{\Psi} \phi(x) \phi(x') \ket{\Psi} \sim_{x \to x'} \bra{\Omega}\phi(x) \phi(x') \ket{\Omega} + \text{subleading terms}.
\end{equation}
Consequently, for any $\ket{\Psi} \in \H,$ the symbol $\phi(x) \ket{\Psi}$ cannot represent a vector in the Hilbert space, because it would need to have infinite norm.
By contrast, smeared objects like $\phi[f]$ can be actual (generally unbounded) operators on $\H$, as the smearing function $f(x)$ can be used to smooth out the singularities, provided that its support is sufficiently nice --- see \cite{wittenenergy} for further discussion.

\begin{figure}[h]
	\centering
	\includegraphics{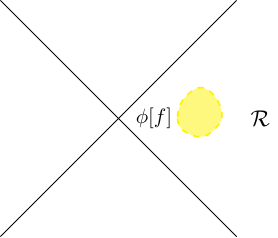}
	\caption{In quantum field theory, a field $\phi(x)$ is an operator-valued distribution, which can be used to create an operator $\phi[f]$ by smearing against a compactly supported function $f$.}
	\label{fig:smearing-1}
\end{figure}

The center of $\A_R$ will consist of observables that we can create by smearing field operators in $\R,$ such that the smeared operator commutes with every other observable that can be made up of field operators in $\R$.
The identity operator has this property, since it is the trivial object constructed by smearing no field operators.
Should we expect any nontrivial central operator to exist?
If the smearing function $f$ contains any points in the interior of $\R$, then $\phi[f]$ will not be central, because then we could consider another smearing function $g$ whose support is timelike separated from the support of $f$, and we would not expect $\phi[f]$ to commute with $\phi[g]$ unless all field operators commute at timelike separation, which would be a very strange property for a field theory to have.
See figure \ref{fig:smearing-2}.
The only way we could hope to construct a central operator in $\A_R$ would be to localize it to the bifurcation surface $\B = \R \cap \L$, where $\L$ is the complementary Rindler wedge.
I.e., there must be some field $\chi(x)$ and some smearing function $b(x)$ localized to $\B$ such that $\chi[b]$ is an operator in the theory.
Even this would not be enough, of course; we would have to verify that $\chi[b]$ commutes with every operator in $\A_R$ whose smearing region includes portions of the Rindler horizon.
But the requirement that there exist a nontrivial central operator $\chi[b]$ localized to the bifurcation surface $\B = \R \cap \L$ is already highly suspect.

\begin{figure}[h]
	\centering
	\includegraphics{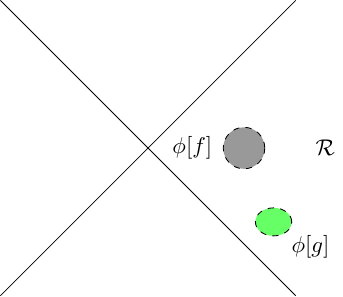}
	\caption{When the supports of $f$ and $g$ are timelike separated, we should not expect $\phi[f]$ and $\phi[g]$ to commute.}
	\label{fig:smearing-2}
\end{figure}

In a theory with charge, like quantum electrodynamics (QED), we might hope that there exists some operator $Q_{\B}$ localized to $\B$ that measures the total flux of the electromagnetic field through $\B$.
In \textit{pure} quantum electrodynamics, with no matter fields to carry charge, we actually can expect such an operator to exist.
This is because if $\B'$ is a surface slightly perturbed from $\B$, then using Stokes' theorem the difference between $Q_{\B}$ and $Q_{\B'}$ can be written as the integral of a current operator over a region lying between $\B$ and $\B'.$
In pure QED, the current operator associated with the electromagnetic field is identically zero, so $Q_{\B}$ and $Q_{\B'}$ represent the exact same physical quantity.
Any smearing of $Q_{\B}$ in directions orthogonal to $\B$ can be written as an integral over some continuous family of deformations $Q_{\B'}$; since each of these is equivalent to $Q_{\B},$ any smearing of $Q_{\B}$ can be written in terms of $Q_{\B}$ itself.
Consequently, after we ``smear $Q_{\B}$ to make it a real operator,'' we will end up with some multiple of $Q_{\B}$, and we can conclude that $Q_{\B}$ was an operator all along.\footnote{This argument was explained to me by Daniel Harlow.}

In a theory with matter fields that carry charge, however, no such argument can be made.
$Q_{\B}$ and $Q_{\B'}$ will no longer represent the same physical quantity, as their difference acts nontrivially on states that have electric charge in a region between $\B$ and $\B'.$
Can we argue that $Q_{\B}$ is not actually an operator in the theory with charge carriers, and consequently argue that the local algebras in a theory with charge carriers are factors?
I don't know how to do this rigorously, but my intuition is as follows.
If $Q_{\B}$ were a true operator in the theory, then we should expect $Q_{\B'}$ to be a true operator as well.
The difference between them, which is expressible as an integral of the current over a region between $\B$ and $\B'$, would then have to be an operator.
But we should not expect the integral of a current over a generic codimension-1 region to be an operator in the theory, since its two-point function is expected to have universal divergences associated to the edge of the region.
This is similar to the lore that a half-sided boost, which can be expressed formally as an integral of the stress-energy tensor over a partial Cauchy slice, is not representable as an operator in any quantum field theory.

\subsection{Every non-factor algebra is a combination of factors}
\label{subsec:factor-decomposition}

Suppose the heuristic arguments of the previous subsection fail, or we are interested in studying the algebraic structure of a theory like pure QED with charge but no charge carriers.
Then we may not be justified in restricting our attention to factors in describing the algebra of operators corresponding to a region of spacetime.
Luckily, there is a theorem due to von Neumann \cite{von1949rings} which states that every von Neumann algebra\footnote{As far as I know, the theorem only holds for von Neumann algebras on separable Hilbert spaces, but these are the ones of interest in physics.} can be expressed as a combination of factors.
To be mathematically precise, every von Neumann algebra can be expressed as a \textit{direct integral} of factors.
The definition of a direct integral is itself fairly technical, and the proof of von Neumann's reduction theorem is even harder to understand, so I will try to provide some intuition.

Suppose we have a von Neumann algebra $\A$ whose center $\Z = \A \cap \A'$ is nontrivial.
Note that $\Z$ is an abelian von Neumann algebra.
It is an interesting fact that on a separable Hilbert space $\H$, every abelian von Neumann algebra is generated by a single Hermitian operator; that is, there exists some Hermitian $T \in \A$ such that $\A = \{T\}'',$ the double commutant of the singleton set $\{T\}.$
Equivalently, as explained in the proof of fact \ref{fact:commutant-vN}, $\A$ consists of all weak limits\footnote{See appendix \ref{app:operator-topologies}.} of polynomials in $T$.
This highly nontrivial theorem was proved by von Neumann in \cite{v1930algebra}, with an English translation available in \cite{muraskin1995neumann}.
It is given as exercise 3f in chapter I.7.3 of \cite{dixmier2011neumann}, where some tips are given as to how it can be proved.
To express $\A$ in terms of factors, the idea is to decompose $\H$ into eigenspaces of $T$, to show that the restriction of $\A$ to each eigenspace is a factor, and to represent $\A$ as the direct sum over these restrictions, i.e., as the set of operators that are block diagonal with respect to the spectral decomposition of $T$.
While this prescription isn't technically accurate in all cases --- since $T$ might have a continuous spectrum, in which case it has no eigenvalues\footnote{This terminology is explained more thoroughly in appendix \ref{app:spectral-theory}. The point is that the spectrum of $T$ consists of all those numbers $\lambda$ for which the operator $T - \lambda$ is not invertible. If the failure of invertibility is because $T - \lambda$ is not injective, then there is an eigenvector $\ket{\lambda}$ with $T \ket{\lambda} = \lambda \ket{\lambda}.$ However, if the failure of invertibility is because $T - \lambda$ is not surjective, then there need not be a corresponding eigenvector, so $\lambda$ need not be an eigenvalue.} --- I find the intuition helpful even in the general case.

The spectrum\footnote{See appendix \ref{app:spectral-theory}.} of $T$, denoted $\sigma(T),$ is some subset of the real line, and the spectral theorem asserts the existence of a measure $\mu$ on $\sigma(T)$ such that $\Z$ is in one-to-one correspondence with the bounded functions of $\sigma(T)$ quotiented by the equivalence relation that sets two functions equal if they differ only on a set of measure zero.
If the spectrum of $T$ happens to be a countable set of discrete eigenvalues, then each point in $\sigma(T)$ has nonzero measure and corresponds to an eigenspace of $T$.
In this case, $\Z$ is the set of block-diagonal operators that act as a multiple of the identity within each eigenspace.
Let us call the eigenspaces $X_j$ and their orthogonal projectors $P_j.$
The Hilbert space $\H$ can be written as the direct sum $\H = \oplus_j X_j.$
It is a general fact about von Neumann algebras that for any orthogonal projection $P \in \A$ whose image is the subspace $X$, the algebra $P \A P$ restricted to $X$ is a von Neumann algebra with center $P \Z P$ --- this is proved, for example, as proposition 43.8 of \cite{conway2000course}.
So each $P_j \A P_j$ is a von Neumann algebra whose center, $P_j \Z P_j,$ is trivial.
That is, each $P_j \A P_j$ is a factor.
One can then show that the original von Neumann algebra $\A$ is equal to the direct sum $\A = \oplus_{j} P_j \A P_j$ --- for every $a \in \A$ we have $a = \sum_{j,k} P_j a P_k,$ and since each $P_j$ is central we have
\begin{equation}
	a = \sum_{j,k} P_j a P_k = \sum_{j,k} P_j P_k a = \sum_{j} P_j P_j a = \sum_{j} P_j a P_j.
\end{equation}

Thus far, I have explained how a general von Neumann algebra can be written as a direct sum of factors in the case where its center is generated by a self-adjoint operator $T$ with discrete spectrum.
In the general case, $T$ might have a completely continuous spectrum, and $\mu$ might assign measure zero to any individual point in the spectrum while assigning nonzero measure to open subsets of the spectrum --- this is the case where $T$ has no eigenvalues at all.
There is still a spectral projection $P_{\omega}$ associated with any open subspace $\omega$ of the spectrum of $T$, which projects onto a subspace $X_{\omega} \subseteq \H.$
We can pick some orthonormal set of spectral projections $P_{\omega_j}$ and decompose $\H$ and $\A$ as $\H = \oplus X_{\omega_j}$ and $\A = \oplus P_{\omega_j} \A P_{\omega_j},$ but there is no reason to expect that the restriction of $P_{\omega_j} \A P_{\omega_j}$ to $X_{\omega_j}$ will be a factor.
In fact, it generally will not be a factor, as if $\omega$ is an open subset of $\sigma(T)$ and $\omega' \subsetneq \omega$ is a proper open subset of $\omega,$ then $P_{\omega'}$ will be a nontrivial operator in the center of $P_{\omega} \A P_{\omega}.$
Intuitively, when $T$ has a continuous spectrum, we can always ``shrink its spectral subspaces further,'' which prevents $P_{\omega} \A P_{\omega}$ from being a factor.

Due to this ``shrinking'' issue, it seems that the only way we could expect to express $\A$ in terms of factors would be to assign, for each $\lambda \in \sigma(T)$, a von Neumann factor $\A_{\lambda}$ and a Hilbert space $X_{\lambda}$ such that, in some appropriate technical sense, we can write the expressions:
\begin{align}
	\H
		& = \int_{\sigma(T)} d\mu\, X_{\lambda}, \\
	\A
		& = \int_{\sigma(T)} d\mu\, \A_{\lambda}.
\end{align}
Ideally, we would like to do this in such a way that the expressions descend naturally to any open subset of the spectrum:
\begin{align}
	X_{\omega}
		& = \int_{\omega}\, d\mu X_{\lambda}, \label{eq:hilbert-space-integral} \\
	P_{\omega} \A P_{\omega}
		& = \int_{\omega}\, d\mu \A_{\lambda}. \label{eq:von-neumann-integral}
\end{align}
We are faced with a pretty big conceptual issue when we try to do this: what Hilbert space $X_{\lambda}$ could we possibly hope to associate with the point $\lambda \in \sigma(T),$ when that point is not an eigenvalue?
There is no eigenspace associated with $\lambda,$ so we can't hope to proceed by choosing the $X_{\lambda}$ Hilbert spaces to be eigenspaces of $T$ (except in the case discussed above where the spectrum of $T$ consists only of discrete eigenvalues).

As far as I can tell, there is no truly ``physical'' answer to this question.
The actual construction is fairly ad hoc; I will sketch a simplified version of it in this paragraph, then provide a reference for further study.
The idea is to remember that due to the spectral theorem, we can map $\H$ unitarily onto the function space $L^2(\sigma(T), \mu).$\footnote{Technically, it may be unitarily equivalent to a function space over several copies of $\sigma(T)$ with slightly different measures, but I will ignore this subtlety --- hence why this is a ``simplified version'' of the construction. See e.g. theorem VII.3 of \cite{reed1980functional}.}
Every vector $\ket{x} \in \H$ then corresponds to some equivalence class of square integrable functions on $\sigma(T)$; for each $\ket{x},$ we will pick some representative $f_x$ within that class.
For any given $\lambda \in \sigma(T),$ we can then define a quadratic form $q_{\lambda}$ on $\H$ given by
\begin{equation}
	q_{\lambda}(x, y) = \bar{f_x(\lambda)} f_y(\lambda).
\end{equation}
This form has all the properties to be an inner product except positive definiteness; it is possible to have $\bar{f_x(\lambda)} f_x(\lambda) = 0$ at a given value of $\lambda$ even if $\ket{x}$ is not the zero vector.
We can define a Hilbert space $X_{\lambda}$ by taking the quotient of $\H$ by all null states of the quadratic form $q_{\lambda},$ which turns $q_{\lambda}$ into an inner product, and then taking the completion of the quotient in that inner product.
The von Neumann algebra $\A$ can be taken to act on $X_{\lambda}$ by passing its action on $\H$ through this quotient; the corresponding algebra can be shown to be a von Neumann algebra $\A_{\lambda},$ whose center $\Z_{\lambda}$ is obtained by passing $\Z$ through the quotient in a similar way.
One then defines continuous integrals of Hilbert spaces and algebras in such a way that the expressions in equations \eqref{eq:hilbert-space-integral} and \eqref{eq:von-neumann-integral} hold, and finds that the center $\Z$ consists exactly of those operators in $\A$ that act as a multiple of the identity on each $X_{\lambda}$.
Consequently, each $\Z_{\lambda}$ consists only of multiples of the identity, and each $\A_{\lambda}$ is a factor.
In doing all of this, it is necessary to be very careful about the usual issues that arise in measure theory --- which functions are measurable, which measurable functions differ only on sets of measure zero, etc.

I would love to tell you that I have a good, intuitive way of thinking about this construction, but I don't.
With a little meditation, you can convince yourself that in the case where the spectrum of $T$ is discrete, the quadratic form $q_{x, y}(\lambda)$ is just given by
\begin{equation}
	q_{x, y}(\lambda) = \bra{x} P_{\lambda} \ket{y},
\end{equation}
with $P_{\lambda}$ the projector onto the $\lambda$-eigenspace of $T$.
From this, it is straightforward to see that the general construction described in the previous paragraph should reduce to the ``direct sum of eigenspaces'' case when $T$ has a discrete spectrum.
I think it is best to take the discrete-spectrum case as the sensible one, and treat the continuous-spectrum generalization as a formal toolkit that ensures we can apply our discrete-spectrum techniques in more general settings.
For the truly intrepid reader, a detailed discussion of this material can be found in chapter II.6 of \cite{dixmier2011neumann}, supplemented by chapters II.1-II.3.

\section{A fundamental example, and the big picture}
\label{sec:big-picture}

Suppose I hand you a quantum system with Hilbert space $\H$ and ask you: ``What operators can you access?''
While the question is partially a philosophical one --- what do I mean by ``access''? --- a reasonable answer would be: ``I have access to $\B(\H)$, the space of bounded operators on $\H$.''
I might then ask, ``Do you have access to any density matrices?''
To this, you would respond ``Yes! There are density matrices in $\B(\H)$.
In fact, given any orthonormal sequence $\ket{\psi_n} \in \H$ and any sequence $p_n \in [0, 1]$ with $\sum_{n} p_n = 1,$ the density matrix $\rho = \sum_{n} p_n \ketbra{\psi_n}$ is in $\B(\H)$.''

So far, our conversation has been both pleasant and simple.
Now, though, I am going to complicate it: I tell you, ``There's another quantum system fifty light years away with Hilbert space $\H'.$ What operators do you have access to?''
Again, this question is philosophical, but since you are now aware of the existence of another quantum system, you might respond, ``I have access to $\B(\H) \otimes 1_{\H'},$ the operators that act as a bounded operator on $\H$ and as the identity on $\H'.$''
I ask you again, ``Do you have access to any density matrices?''
This question is suddenly harder to answer! If $1_{\H'}$ has finite dimension $d$, then the answer is yes, because for any density matrix $\rho \in \B(\H)$ you have access to the operator
\begin{equation}
	\frac{\rho}{d} \otimes 1_{\H'} = \rho \otimes \frac{1_{\H'}}{d},
\end{equation}
which is a density matrix.
But if $\H'$ is infinite dimensional, then the algebra $\B(\H) \otimes 1_{\H'}$ contains no density matrices, because there is no way to normalize $\rho \otimes 1_{\H'}$ to make it a density matrix on the combined system.
I.e., there is no way to rescale this operator to make it have unit trace.

What has happened?
Before I told you about the existence of an infinite-dimensional $\H',$ you would have said that your accessible algebra of observables contained density matrices, but once I have told you about the existence of $\H'$, you are forced to say that your algebra of observables contains no density matrices.
It might seem like we are getting caught up in semantics, but this puzzle is at the core of the type classification of von Neumann algebras and their role in quantum physics.
Its resolution is that while $\rho \otimes 1_{\H'}$ might not be a density matrix for the combined system $\H \otimes \H'$ in the case that $\H'$ is infinite-dimensional, it is an \textit{effective} density matrix as far as you are concerned, because you only have access to the observables in $\B(\H) \otimes 1_{\H'}.$
Within this algebra of observables, there is a consistent way to assign expectation values $\langle A \otimes 1_{\H'} \rangle_{\rho \otimes 1_{\H'}}$ that gives $\rho \otimes 1_{\H'}$ all the properties of a quantum state.
For any bounded operator $A \otimes 1_{\H'},$ we define its \textit{effective} expectation value in the ``state'' $\rho \otimes 1_{\H'}$ in the obvious way:
\begin{equation}
	\langle A \otimes 1_{\H'} \rangle_{\rho \otimes 1_{\H'}} \equiv \tr_{\H} (\rho A).
\end{equation}

The slogan I would like to advance is that while the operator $\rho \otimes 1_{\H'}$ is not a true density matrix, it is a \textit{renormalizable} density matrix with respect to the algebra $\B(\H) \otimes 1_{\H'}$ --- within this algebra, there is a consistent way of assigning expectation values to $\rho \otimes 1_{\H'}$ that gives it the properties of a quantum state.
Furthermore, this renormalization procedure is not special to the particular operator $\rho \otimes 1_{\H'}$: for any other $\H$-density matrix $\sigma,$ we assign expectation values to the operator $\sigma \otimes 1_{\H'}$ in a completely analogous way:
\begin{equation}
	\langle A \otimes 1_{\H'} \rangle_{\sigma \otimes 1_{\H'}} \equiv \tr_{\H} (\sigma A).
\end{equation}

Now, let us move beyond the special case where we are given an access to an algebra of the form $\B(\H) \otimes 1_{\H'},$ and suppose that we are given access to some general von Neumann algebra $\A$.
This could be, for example, the von Neumann algebra associated to a subregion in the vacuum sector of a quantum field theory.
Within this algebra $\A$ there is a set of positive operators that we will denote $\A_+.$
Are any of these operators density matrices?
Maybe not.
But perhaps, as above, there are some operators in $\A_+$ that are effective density matrices for the observables in $\A$.
We can ask: does there exist some consistent renormalization scheme on $\A$ that makes these operators in $\A_+$ into quantum states?

We will discuss the details of what is meant by a ``renormalization scheme'' in the following sections.
In particular, we will see that in a von Neumann factor, there is really only one consistent renormalization scheme that acts on every operator in the same way.
With that in mind, I now claim that the type classification of von Neumann factors can be thought of in the following terms:
\begin{itemize}
	\item A factor $\A$ is \textbf{type I} if its renormalization scheme turns some (possibly all) operators in $\A_+$ into pure states, and some into mixed states.
	I.e., it contains ``renormalizable pure and mixed states.''
	\item A factor $\A$ is \textbf{type II} if its renormalization scheme turns some (possibly all) operators in $\A_+$ into mixed states, but none into pure states.
	I.e., it contains ``renormalizable mixed states but no renormalizable pure states.''
	\item A factor $\A$ is \textbf{type III} if, even after renormalization, it contains no density operators.
	I.e., it contains ``no renormalizable states.''
\end{itemize}

A complementary distinction can be made:
\begin{itemize}
	\item A factor $\A$ is \textbf{finite} if its renormalization scheme makes every operator in $\A_+$ into a density matrix.
	I.e., every positive operator in $\A$ is a renormalizable state.
	\item A factor $\A$ is \textbf{infinite} if its renormalization scheme leaves at least one operator in $\A_+$ unnormalizable.
	I.e., it contains at least one positive operator that is not a renormalizable state.
\end{itemize}

Note that based on the above definitions, every type III factor is infinite, while type I and II factors can be either finite or infinite.

Let us conclude this section with some terminology that will be elaborated in the following sections.
A finite factor of type II is said to be \textbf{type II$_1$}; an infinite factor of type II is said to be \textbf{type II$_{\infty}.$}
A finite factor of type I is said to be \textbf{type I$_n$}, where $n$ is an integer that encodes some additional information about the factor.\footnote{I give a precise definition of $n$ in subsection \ref{subsec:algebraic-mixed-states}.}
An infinite factor of type I is said to be \textbf{type I$_{\infty}.$}

Several useful examples of factors of various types can be found in \cite[section 6]{witten2018aps}; in particular, to calibrate your intuition for algebra types before embarking on a careful study of their definitions, it may be useful to check \cite[section 6.3]{witten2018aps} to see an example of a type II$_1$ algebra that arises in the thermodynamic limit of two spin chains that are maximally entangled with one other.

\section{The algebraic classification of factors}
\label{sec:projector-types}

For ease of reading, I have split this section into three subsections.
The first explains why, in our quest to understand the ``renormalizable density matrices'' in a von Neumann algebra $\A$, we should start by studying certain algebraic properties of the projectors in $\A.$
The second subsection details some essential features of this algebraic structure; proofs are given for completeness, but it is fine to read only the statements without understanding the proofs.
The third subsection reframes the heuristic type classification of section \ref{sec:big-picture} in terms of the structure of the projectors in $\A$.

\subsection{The relative dimension of projectors}

How do we know whether a positive operator is a density matrix?
There are a few ways of answering this.
The most naive would be to say, ``a positive operator is a density matrix if it has finite trace.''\footnote{You might want to say that it must have trace equal to one, but this isn't really necessary, as for an unnormalized density matrix we can define expectation values using the formula $\langle A \rangle_{\rho} = \tr(\rho A)/\tr(\rho).$}
We will come back to this idea in section \ref{sec:traces}, but first it will be helpful for us to cook up a more algebraic test for whether a given positive operator is a density matrix.

Every positive operator can be approximated arbitrarily well\footnote{Here ``approximated'' is meant in the sense of the norm topology. Using spectral theory, a bounded positive operator can be thought of as the identity function on a compact measure space, corresponding to the spectrum of the operator. Linear combinations of spectral projections are simple functions on the spectrum --- ``simple'' here is a technical term, meaning a finite linear combination of indicator functions. The statement that a positive operator can be approximated arbitrarily well with simple functions is then the standard measure theory fact that simple functions are dense in any $L^{\infty}$ space. See appendix \ref{app:spectral-theory}.} by positive linear combinations of its spectral projections, so we should ask first what it means for a projection to be a density matrix, and then ask about whether the property ``being a density matrix'' is preserved under linear combinations and limits.
The answer is simple: a projection $P$ is a density matrix if and only if it has finite rank.
Given the example studied in section \ref{sec:big-picture}, though, we know that a projection can be an ``effective'' density matrix even if it is not finite-rank.
If $P$ is a finite-rank projector on the Hilbert space $\H$, and $\H'$ is an infinite-dimensional Hilbert space, then the projector $P \otimes 1_{\H'}$ is not a finite-rank projector, but it is effectively finite-rank from the perspective of the algebra $\B(\H) \otimes 1_{\H'}.$

This motivates us to come up with a definition of what it means for $P$ to be finite-rank with respect to some von Neumann algebra $\A$ containing $P$.
Since we want the definition to reference the structure of $\A$, it must be inherently algebraic.
A good algebraic characterization of finite-rank projectors is that a projector $P \in \B(\H)$ is finite-rank if and only if every proper sub-projection $P' \lneq P$ has cardinality strictly different from the cardinality of $P$.
(An infinite-rank projection does not have this property, since throwing away a one-dimensional subspace results in an infinite-rank, proper sub-projection.)
A good algebraic characterization of the cardinality of a projector is that two projectors $P$ and $Q$ project onto subspaces of the same dimension if and only if there exists a partial isometry $V$ with initial space $P \H$ and final space $Q \H$, i.e., if there exists a partial isometry with $V^* V = P$ and $V V^* = Q.$

To avoid confusion, we will now start reserving the term ``rank of a projector'' for the literal dimension of the projector's image, and use the term ``dimension of a projector'' to refer to a more abstract concept that can depend on an algebra containing that projector.
The considerations of the preceding paragraph motivate the following definition:
\begin{definition} \label{def:projector-equivalence}
	Let $\A$ be a von Neumann algebra.
	Two projections $P, Q \in \A$ are said to have \textbf{the same dimension relative to $\A$} if there exists a partial isometry $V \in \A$ with $V^* V = P, V V^* = Q.$
	Often we will just say $P$ and $Q$ are \textbf{equivalent},\footnote{It is straightforward to check that this is actually an equivalence relation, i.e., that we have (i) $P \sim P$, (ii) $P \sim Q$ and $Q \sim R$ implies $P \sim R$, (iii) $P \sim Q$ implies $Q \sim P$.} without explicitly referencing $\A$ or $V$, and write $P \sim Q.$
	
	A projection $P$ in a von Neumann algebra $\A$ is \textbf{finite-dimensional relative to $\A$}, or simply \textbf{finite}, if every proper subprojection $Q \in \A, Q \lneq P$ is inequivalent to $P$.
	Conversely, $P$ is \textbf{infinite-dimensional relative to $\A$}, or simply \textbf{infinite}, if there exists a projector $Q \in \A$ with $Q \lneq P$ and $Q \sim P.$
\end{definition}

Notice that every infinite projector $P \in \A$ must be infinite-rank in the Hilbert space sense; this is because if $P$ is infinite in $\A$ it is clearly infinite in $\B(\H)$, and the algebraic definitions of ``finite'' versus ``infinite'' in $\B(\H)$ correspond to the standard notions of finite and infinite rank.
However, there can exist projectors of infinite rank that are finite relative to a von Neumann algebra $\A$.
One way this can happen is if $P \in \A$ is an infinite-rank projector, but none of its infinite-rank proper subprojectors are contained in $\A$, so that the algebra $\A$ has no way of seeing that $P$ has infinite rank.
This happens, e.g., for the trivial von Neumann algebra consisting of scalar multiples of the identity in an infinite-dimensional Hilbert space.
Alternatively, even if $\A$ contains an infinite-rank proper subprojector $Q \lneq P$, it might not contain any partial isometry $V$ with $V^* V = P$ and $V V^* = Q,$ so $\A$ cannot see that $Q$ and $P$ have the same rank!
In either of these cases, $P$ will be finite relative to $\A$, even though it is an infinite-rank projector.

We have a good definition now of what it means for a projector $P$ to be ``effectively finite dimensional'' with respect to $\A$, but we still haven't shown that this definition satisfies our original goal of characterizing the projectors $P$ that can be consistently renormalized into maximally mixed states on subspaces of $\H$.
To do this, we will need to introduce traces; this will be done in section \ref{sec:traces}.
First, though, I will explain some general properties of the projector equivalence relation that will be very important for developing the theory.
I will then give a preliminary classification of factors in terms of what kinds of projectors they contain, which will be physically justified in section \ref{sec:traces} when we show that finite projectors can be renormalized consistently into density matrices.

\subsection{Properties of relative dimension}
\label{sec:relative-dimension}

This subsection is written in ``definition-lemma-theorem-proof'' format so that you can just skim the statements if you like, without needing to read all of the proofs.

\begin{lemma}[Additivity of equivalence] \label{lem:additivity-lemma}
	If $\{P_j\}$ and $\{Q_j\}$ are pairwise-orthogonal sequences of projectors in $\A$ with $P_j \sim Q_j,$ then the projectors $\sum_{j} P_j$ and $\sum_j Q_j$ are in $\A$ and are equivalent.
\end{lemma}
\begin{proof}
	One might first worry about the definition of the operator $\sum_j P_j$; it is defined as the projector onto the direct sum $\oplus_{j} P_j \H.$
	It is straightforward to check that for any $\ket{x}\in \H$, we have
	\begin{equation} \label{eq:projector-sum-strong-convergence}
		\left( \sum_j P_j \right) \ket{x} = \sum_j P_j \ket{x}.
	\end{equation}
	At a technical level, this means that the operator $\sum_j P_j$ is the limit of its partial sums in the strong operator topology.
	As mentioned in appendix \ref{app:operator-topologies}, von Neumann algebras are closed in the strong operator topology, so $\sum_j P_j$ and analogously $\sum_j Q_j$ are projectors in $\A$.\footnote{We could also show that $\sum_j P_j$ and $\sum_j Q_j$ are projectors in $\A$ using fact \ref{fact:projector-infima-and-suprema}.}

	Given the assumptions of the lemma, we know there exists a family of partial isometries $V_j$ with $V_j^* V_j = P_j$ and $V_j V_j^* = Q_j.$
	We define the operator $\sum_{j} V_j$ by having it act on the vector $\ket{x}$ as
	\begin{equation}
		\left(\sum_j V_j \right) \ket{x} = \sum_j V_j \ket{x}.
	\end{equation}
	In order for this to be well defined, we need the sum $\sum_j V_j \ket{x}$ to converge for all $\ket{x}.$
	But since the $Q_j$ projectors are pairwise orthogonal, the final spaces of the $V_j$ partial isometries are pairwise orthogonal.
	The sum of a pairwise-orthogonal set of vectors converges if and only if the sum of the squared-norms converges, so we need only check the condition
	\begin{equation}
		\sum_j \braket{V_j x}{V_j x} < \infty.
	\end{equation}
	But since we have $V_j^* V_j = P_j =  P_j^* P_j,$ this is the same as checking that the sum $\sum_j \braket{P_j x}{P_j x}$ converges, and this holds as a consequence of equation \eqref{eq:projector-sum-strong-convergence}.
	By a similar argument, one can show that the operator
	\begin{equation}
		\left(\sum_j V_j^* \right) \ket{x} = \sum_j V_j^* \ket{x}
	\end{equation}
	is well defined, and that we have $\left(\sum_j V_j\right)^* = \sum_j V_j^*.$
	It is then not too hard to show the identities
	\begin{equation}
		\left(\sum_j V_j \right)^* \left(\sum_j V_j\right) = P.
	\end{equation}
	\begin{equation}
		\left(\sum_j V_j \right) \left(\sum_j V_j\right)^* = Q.
	\end{equation}
	At a technical level, this is done by noting that the series in these expressions are limits of partial sums in the strong operator topology, and then using the fact that left- or right-multiplication by a fixed operator is a continuous function in the strong operator topology.
	That is, we have
	\begin{align}
		\left(\sum_j V_j \right)^* \left(\sum_j V_j\right)
			& = \left(\sum_j V_j^* \right) \left(\sum_j V_j\right) \nonumber \\
			& = \sum_j V_j^* \left( \sum_k V_k \right) \nonumber \\
			& = \sum_j \sum_k \left( V_j^* V_k \right) \nonumber \\
			& = \sum_j V_j^* V_j \nonumber \\
			& = P.
	\end{align}
\end{proof}

\begin{definition}
	If $P, Q$ are projectors in a von Neumann algebra $\A$, then we write $P \precsim Q$ if there is a projector $R \in \A$ with $P \sim R \leq Q.$
\end{definition}

\begin{theorem}[The comparison theorem] \label{thm:comparison}
	In any von Neumann algebra $\A$, the relation $\precsim$ defined above is a partial order:
	\begin{itemize}
		\item For any $P$, we have $P \precsim P.$
		\item For any $P, Q, R$ with $P \precsim Q$ and $Q \precsim R,$ we have $P \precsim R$.
		\item For any $P, Q$ with $P \precsim Q$ and $Q \precsim P,$ we have $P \sim Q.$
	\end{itemize}
	
	If $\A$ is a factor, then $\precsim$ is a total order --- for any two projectors $P$ and $Q$, we have either $P \precsim Q$ or $Q \precsim P$ (or both).
\end{theorem}
\begin{proof}[Sketch of proof]
	The first two bullet points aren't too hard to show directly from the definitions.
	The third bullet point is nontrivial, and the proof is clever; a readable account is given as proposition 1.3 in chapter V.1 of \cite{takesaki2001theory}.
	
	The final statement, that $\A$ being a factor implies $P \precsim Q$ or $Q \precsim P,$ uses a lemma that in a factor, for any two nonzero projections $P, Q$ there exist some projections $P', Q'$ with $P \geq P' \sim Q' \leq Q$; this can be found e.g. as lemma 1.7 in chapter V.1 of \cite{takesaki2001theory}.
	One then uses Zorn's lemma to construct some projectors $\tilde{P}, \tilde{Q}$ satisfying $P \geq \tilde{P} \sim \tilde{Q} \leq Q$ that are maximal with respect to this property.
	If we have $P = \tilde{P}$ then we have $P \precsim Q$, and similarly if we have $Q = \tilde{Q}$ we have $Q \precsim P.$
	In either case, the theorem would be proved.
	The third possibility is that both $P - \tilde{P}$ and $Q - \tilde{Q}$ are nonzero.
	But in this case, applying lemma 1.7 of \cite{takesaki2001theory} again, we can product nonzero projectors $R \leq P - \tilde{P}$ and $S \leq Q - \tilde{Q}$ with $R \sim S.$
	In this case, by applying the additivity lemma (\ref{lem:additivity-lemma}), we can show that $\tilde{P} + R$ and $\tilde{Q} + S$ satisfy $P \geq \tilde{P} + R \sim \tilde{Q} + S \leq Q,$ which contradicts maximality of the pair $(\tilde{P}, \tilde{Q})$.
\end{proof}

\begin{prop}[Finiteness is preserved under equivalence] \label{prop:finiteness-equivalence}
	If $P$ and $Q$ are projectors in $\A$ with $P \sim Q$ and $P$ is finite, then $Q$ is finite.
\end{prop}
\begin{proof}
	Suppose, toward contradiction, that $Q$ is infinite.
	Then there exists a projection $R \in \A$ with $R \lneq Q$ and $R \sim Q.$
	Let $V$ be a partial isometry in $\A$ furnishing the equivalence between $Q$ and $P$, satisfying $V^* V = Q$ and $VV^* = P.$
	Then $V R V^*$ is a projector with $V R V^* \lneq P$, and the partial isometry $V R$ furnishes the equivalence $R \sim V R V^*,$ since we have $V R V^* = (V R) (V R)^*$ and $R = (V R)^* (V R).$
	By transitivity of equivalence, we have $V R V^* \sim R \sim Q \sim P,$ so $V R V^* \lneq P$ and $V R V^* \sim P,$ which contradicts finiteness of $P.$
\end{proof}

\begin{prop}[Finiteness is hereditary]
	If $\A$ is a von Neumann algebra and $P \in \A$ is finite, then any $Q \in \A$ with $Q \precsim P$ is finite.
\end{prop}
\begin{proof}
	In light of proposition \ref{prop:finiteness-equivalence}, it suffices to prove the theorem in the case $Q \leq P,$ since combining with proposition \ref{prop:finiteness-equivalence} yields the general case $Q \precsim P.$
	Suppose, toward contradiction, that $Q$ is infinite, so there exists $R \in \A$ with $R \lneq Q$ and $R \sim Q.$
	By the additivity lemma (lemma \ref{lem:additivity-lemma}), we have
	\begin{equation}
		R + (P - Q) \sim Q + (P - Q) = P.
	\end{equation}
	So $P$ is equivalent to $R + (P - Q),$ but $R \lneq Q$ implies $R + (P - Q) \lneq P,$ which contradicts the finiteness of $P$.
\end{proof}

\begin{definition}
	A factor is said to be \textbf{finite} if all of its projections are finite.
	It is \textbf{infinite} if it contains at least one infinite projection.
	
	In light of the preceding proposition, it suffices to check the identity operator --- a factor is finite if the identity operator is a finite projection, and infinite if the identity operator is an infinite projection.
	This is because for every $P \in \A$, we have $P \leq 1,$ so $P$ being infinite implies that $1$ is infinite and $1$ being finite implies that $P$ is finite.
\end{definition}

\begin{remark}
	We showed in proposition \ref{prop:finiteness-equivalence} that the finiteness or infiniteness of a projector is preserved under the equivalence relation introduced in definition \ref{def:projector-equivalence}.
	Is the converse statement true?
	That is, are any two finite projectors equivalent?
	Are any two infinite projectors equivalent?
	
	The answer to one of these questions is no --- there will generally exist multiple inequivalent finite projectors.
	By contrast, in a factor, all infinite projectors are equivalent.
	This statement will be important for our understanding of the renormalization schemes developed in section \ref{sec:traces}.
	Essentially it will tell us that we can consistently assign renormalized dimensions to projectors such that the set of equivalence classes is in one-to-one correspondence with the set of renormalized dimensions; in order for this to be possible, all infinite projectors, for which we will see the only reasonable renormalized dimension is $+\infty,$ must be equivalent.
\end{remark}

\begin{prop} \label{prop:infinite-projectors-equivalent}
	If $\A$ is a factor, then any two infinite projectors are equivalent.
\end{prop}
\begin{proof}
	In order for this statement to be nontrivial, $\A$ must contain at least one infinite projector, in which case the identity operator is an infinite projector.
	So it suffices to show that if the identity operator is infinite, then any infinite $P \in \A$ satisfies $P \sim 1.$
	Since $P$ is infinite, there exists some projector $P_1 \lneq P$ with $P_1 \sim P.$
	That is, there exists a partial isometry $V$ with $V V^*= P$ and $V^* V = P_1.$
	Consider now the sequence of operators $P_n = (V^*)^n (V)^n.$
	These operators are Hermitian, and they actually square to the identity, which can be seen via repeated application of the partial isometry identity $V V^* V = V$ and the projector nesting identity $P_1 P = P_1.$
	In fact, by applying these identities, one can also show the useful identities $V^n (V^*)^n = P$ and $V P = V.$
	
	Furthermore, the sequence $P_n$ is strictly decreasing; i.e., we have $P_1 \gneq P_2 \gneq P_3 \gneq \dots.$
	To show this, since $P_n$ is a sequence of nested projectors, it suffices to show that the kernel of $P_{n}$ is strictly contained within the kernel of $P_{n+1}.$
	But it is a general fact about bounded operators that the kernel of an operator $T$ is equal to the kernel of $T^* T,$ so we have
	\begin{equation}
		\ker(P_n) = \ker((V^*)^n V^n) = \ker(V^n).
	\end{equation}
	So we want to show $\ker(V^n) \subsetneq \ker(V^{n+1}).$
	The fact that this is an inclusion is obvious.
	To show that it is strict, we use the identity
	\begin{equation}
		V \ket{\psi} = V P \ket{\psi} = V^{n+1} (V^*)^n \ket{\psi}.
	\end{equation}
	So for any $\ket{\psi}$ in the kernel of $V$ (which is the same as the kernel of $P_1$), the state $(V^*)^n \ket{\psi}$ is in the kernel of $V^{n+1}.$
	So if we had $\ker(V^{n+1}) = \ker(V^n),$ then for any $\ket{\psi} \in \ker(V)$ we would have $V^n (V^*)^n \ket{\psi} = 0,$ i.e. $P \ket{\psi} = 0,$ which would give us the inclusion $\ker(P_1) \subseteq \ker(P).$
	This is forbidden by the assumption $P_1 \lneq P,$ so we can conclude that the sequence $P_n$ is strictly decreasing.
	
	Once we have this, we can define a new sequence of projectors $Q_n$ by $Q_1 = P - P_1,$ $Q_2 = P_1 - P_2,$ and so on.
	Taking $P_{\infty} = \lim_{n} P_n$ to be the projector onto $\cap_j P_j \H,$ we obtain the identity
	\begin{equation}
		P = P_{\infty} + \sum_{n} Q_n.
	\end{equation}
	Furthermore, all of the various $Q_n$ projectors are equivalent to one another, since one can check that the partial isometry $Q_n V$ satisfies $(Q_n V) (Q_n V)^* = Q_n$ and $(Q_n V)^* (Q_n V) = Q_{n+1}.$
	Now, we can use Zorn's lemma and the separability of $\H$ to produce a maximal sequence of pairwise-orthogonal projectors $\{R_n\}$ such that each $R_n$ satisfies $R_n \precsim Q_1.$
	Maximality tells us that we must have $\sum_{n} R_n = 1,$ since otherwise theorem \ref{thm:comparison} 	tells us that we have either $Q_1 \precsim 1 - \sum_n R_n$ or $1 - \sum_n R_n \precsim Q_1,$ and either of these possibilities contradicts maximality.
	Since we have $R_n \precsim Q_n$ for each $n,$ we can apply the additivity lemma (lemma \ref{lem:additivity-lemma}) to obtain
	\begin{equation}
		1 = \sum_{n} R_n \precsim \sum_{n} Q_n \leq P_{\infty} + \sum_{n} Q_n = P.
	\end{equation}
	So we have shown $1 \precsim P,$ and we automatically have $P \precsim 1,$ so theorem \ref{thm:comparison} tells us that we have $P \sim 1,$ which is what we wanted to prove all along.
\end{proof}

\subsection{Projectors, pure states, and mixed states}
\label{subsec:algebraic-mixed-states}

Several times in this section, I have foreshadowed that in section \ref{sec:traces} we will construct renormalization schemes that turn certain positive operators into density matrices.
Based on the discussion so far, the details of these renormalization schemes are not apparent.
The details will be explained in section \ref{sec:traces}; for now, I will simply tell you that after renormalization, every nonzero finite projector will become a density matrix, and every infinite projector will remain unnormalizable.
In fact, the renormalization schemes we construct in section \ref{sec:traces} will take this idea as their starting point --- they will assign an effective dimension to each projector, such that finite projectors take on a finite effective dimension while infinite projectors remain infinite-dimensional even after renormalization.
The question of whether a generic positive operator can be renormalized to a density matrix is then a simple question of whether its spectral decomposition consists of finite projectors with sufficiently tame coefficients for $\rho$ to have finite renormalized trace.

Before we tackle the technical details of renormalization, let us take this general principle --- ``nonzero finite projectors are renormalizable density matrices, infinite projectors are unrenormalizable density matrices, and all renormalizable density matrices are made up of finite projectors'' --- and study its consequences for the type classification of von Neumann algebras.
This way of thinking leads us to a more concrete way of framing the classification of factors indicated at the end of section \ref{sec:big-picture}.
There, I said that a factor is type III if it contains no renormalizable density operators.
In our present language, this means that a factor is type III if it contains no nonzero finite projectors.
By contrast, a factor that contains  nonzero finite projectors must be either type I or II.
In section \ref{sec:big-picture}, I said that the difference between these two types is that a type I factor contains pure states, while a type II factor does not.
To make this distinction in our present language, we need to decide what it means for a finite projector to be pure or mixed.
The natural definition is that a nonzero finite projector is pure if it is minimal --- i.e. $P \in \A$ will correspond to a pure state after renormalization if there exists no nonzero $Q \in \A$ with $Q \lneq P.$
We will see in section \ref{sec:traces} that these projectors actually do turn into pure states under renormalization, but the intuition is obvious: a pure state is one that cannot be written as a convex combination of other states, and minimal projectors are the ones with this property.
So a factor with nonzero minimal projectors will contain pure states after renormalization, while a factor with no minimal projectors will still have density matrices after renormalization, but none of them will be pure.

Taking these ideas together, we are led to the following refined type classification of factors:
\begin{itemize}
	\item A factor $\A$ is \textbf{type I} if it contains a nonzero minimal projector. (A minimal projector is necessarily finite; this follows immediately from the definition of a finite projector.)
	\item A factor $\A$ is \textbf{type II} if it contains nonzero finite projectors, but no nonzero minimal projectors.
	\item A factor $\A$ is \textbf{type III} if it contains no nonzero finite projections.
\end{itemize}
Furthermore, we have:
\begin{itemize}
	\item A factor $\A$ is \textbf{finite} if all projections are finite, equivalently if the identity operator is finite.
	\item A factor $\A$ is \textbf{infinite} if at least one projection is infinite, equivalently if the identity operator is infinite.
\end{itemize}
Before proceeding, the reader may wish to cross-check these bullets against those at the end of section \ref{sec:big-picture}, to see how the ideas of ``finite vs. infinite projector'' correspond to the ideas of ``renormalizable vs. unrenormalizable density matrix.''
The reader may also wish to consult \cite[section 6]{witten2018aps} for helpful examples of algebras of the various types that arise in the thermodynamic limit of a pair of spin chains with a particular pattern of entanglement.

At the end of section \ref{sec:big-picture}, I indicated that finite type II factors are said to be of \textbf{type II$_1$}, that infinite type II factors are said to be of \textbf{type II$_{\infty},$} that infinite type I factors are said to be of \textbf{type I$_{\infty},$} and that finite type I factors are said to be of \textbf{type I$_n$}.
I have not yet indicated the origins of this mysterious integer $n.$
In our present language, $n$ is the cardinality of a maximal set of pairwise-orthogonal minimal projections in $\A$.
Such a set always exists, which can be shown using induction or using Zorn's lemma, and the cardinality $n$ is well defined.
We will see in section \ref{sec:standard-forms} that a type I factor can always be put in ``standard form'' as an algebra of the form $\B(\H) \otimes 1_{\H'},$ and $n$ can equivalently be taken as the dimension of the Hilbert space $\H$.

\section{Traces and renormalization}
\label{sec:traces}

Throughout these notes so far, I have claimed that the type classification of von Neumann factors can be thought of in terms of what ``renormalization schemes'' they allow.
The goal of this section is to explain what those schemes actually are.
In subsection \ref{subsec:def-of-trace}, I explain why a renormalization scheme on a von Neumann factor should be thought of in terms of a ``trace'' on that factor.
I list the basic properties of renormalized traces on factors, and explain how these mathematical properties justify the heuristic definition of factor types given at the end of section \ref{sec:big-picture}.
In subsection \ref{subsec:trace-construction}, I sketch the proof that every factor has a unique renormalized trace up to rescaling.

\subsection{What is a trace? How does it renormalize states?}
\label{subsec:def-of-trace}

Suppose we have a von Neumann factor $\A$, whose subset of positive operators is denoted $\A_+.$
What does it mean to renormalize $\A_+$?
How do we produce a general renormalization scheme so that certain operators in $\A_+$ become legitimate density matrices?
Our hint for how to proceed comes from the example discussed in section \ref{sec:big-picture}.
There, it was explained that even in the case that $\H'$ is infinite-dimensional, the operator $\rho_{\H} \otimes 1_{\H'}$ can be thought of as an effective density matrix for the algebra $\B(\H) \otimes 1_{\H'},$ since all expectation values of that algebra with respect to $\rho_{\H} \otimes 1_{\H'}$ can be defined in a way that gives this operator the statistics of a density matrix on $\B(\H).$
So what we really need is a way to define expectation values of operators in $\A$ with respect to an operator in $\A_+.$
Furthermore, we would like this rule for assigning expectation values to be universal.
That is, we want to renormalize all operators in $\A_+$ simultaneously, so our rule for assigning expectation values shouldn't be allowed to vary completely arbitrarily for different operators in $\A_+.$

Our hint for how to proceed is that before doing renormalization, we say $\rho \in \A_+$ is a density matrix if it satisfies $\tr(\rho) = 1,$ and we then assign expectation values via the rule $\langle T \rangle_{\rho} = \tr(\rho T) / \tr(\rho).$
So we can perform renormalization by coming up with a ``renormalized trace'' on $\A_+$ that has all of the right properties to be used for assigning expectation values.
Mathematicians use the following definition.
\begin{definition} \label{def:trace}
	A \textbf{trace} on the von Neumann algebra $\A$ is a map $\tau : \A_+ \to [0, \infty]$ satisfying:
	\begin{enumerate}[(i)]
		\item $\tau(\lambda T) = \lambda \tau(T)$ for all $T \in \A_+$ and all $\lambda \geq 0.$
		\item $\tau(T + S) = \tau(T) + \tau(S)$ for all $T, S \in \A_+.$ (Note that the sum of positive operators is always positive, so this property makes sense.)
		\item $\tau(U T U^*) = \tau(T)$ for all $T \in \A_+$ and all unitary $U \in \A.$
	\end{enumerate}
	
	Note that $\tau$ can assign infinite values to some operators; for example, the standard Hilbert space trace ``$\tr$'' is a trace on $\B(\H)$, but it assigns infinite trace to the identity in the case that $\H$ is infinite dimensional.
\end{definition}
These properties all certainly make sense to require.
Properties (i) and (ii) just say that the trace $\tau$ is linear in the appropriate sense for a map that is only defined on positive operators.
Property (iii) says that the trace is unchanged under conjugation by unitaries in $\A$, which means that our trace is ``basis independent'' with respect to $\A.$
Given properties (i) and (ii), property (iii) is actually equivalent to the requirement that for any $O \in \A,$ we have $\tau(O O^*) = \tau(O^* O)$; the proof of this statement is nontrivial, and is given in appendix \ref{app:trace-math}.

Many wonderful properties of $\tau$ can be proven just from definition \ref{def:trace}.
Proofs are given in appendix \ref{app:trace-math}; these properties are summarized in the following list.
\begin{enumerate}[(i)]
	\item Given a trace $\tau$ on the von Neumann algebra $\A,$ the set $\A_1$ of operators $T \in \A$ for which $\tau(|T|)$ is finite forms an \textit{ideal}.
	This means that sums and scalar products of operators in $\A_1$ remain in $\A_1,$ and for any $S \in \A$ and $T \in \A_1,$ we have $S T \in \A_1$ and $TS \in \A_1.$
	The operators in $\A_1$ are called the \textbf{trace-class operators} of $\A$ with respect to $\tau.$
	\item The trace $\tau$ extends naturally to a linear functional $\tau : \A_1 \to \comps$; for a given $T \in \A_1,$ one decomposes $T$ as a linear combination of four positive operators (the positive and negative parts of its Hermitian and antihermitian parts), then defines $\tau$ on $\A_1$ by linearity.
	\item The extended trace is cyclic on $\A_1$: if $T$ is in $\A_1$ and $S$ is in $\A$, then we have $\tau(T S) = \tau(S T).$
\end{enumerate}
These properties are very valuable, because they give us everything we need to assign expectation values.
If a positive operator $\rho \in \A_+$ is trace-class --- that is, if $\tau(\rho)$ is finite --- then for \textit{any} $T \in \A,$ the number $\tau(\rho T)$ is well defined, finite, and depends linearly on $T$.
So for any $\rho \in \A_+$ with $\tau(\rho)$ finite, we get a rule for assigning expectation values:
\begin{equation} \label{eq:expectation-rule}
	\langle T \rangle_{\rho} = \frac{\tau(\rho T)}{\tau(\rho)}.
\end{equation}
We see now why it makes sense to call $\tau$ a renormalization scheme --- it gives us a rule for assigning expectation values to any positive operator $\rho$ with $\tau(\rho)$ finite, even if the Hilbert space trace $\tr(\rho)$ is infinite!

Definition \ref{def:trace} is perfectly good for assigning expectation values to certain operators $\rho \in \A_+,$ but we will need to put a few more conditions on the map $\tau$ in order for those expectation values to be physical.
One condition is that while we might hope for our renormalized trace to make some unnormalizable density operators normalizable, we don't want it to annihilate any density matrices --- that is, we don't want to have $\tau(\rho) = 0$ for nonzero $\rho.$
Another condition, motivated by the discussion in section \ref{subsec:algebraic-mixed-states}, is that we want our trace to make algebraically finite projectors into normalizable states; that is, if $P \in \A$ is an algebraically finite projector, we want $\tau(P) < \infty.$
A final condition is that we would like our trace $\tau$ to be continuous in some appropriate sense. It is a little confusing what the right definition of continuity is given that the image of $\tau$ can include $+\infty$.
A good definition is to require that if we have some family $\{\rho_{\alpha}\}$ of positive operators with supremum $\rho = \sup_{\alpha} \rho_{\alpha},$ then we have $\tau(\rho) = \sup_{\alpha} \tau(\rho_{\alpha}).$
Another way to think about this condition is that it means the behavior of $\tau$ on a generic positive operator is completely determined by its behavior on projectors, since every positive operator can be written as a supremum over linear combinations of its spectral projectors.
This is explained further in subsection \ref{subsec:trace-construction}; the point is that we can think of $\tau$ as assigning an effective renormalized dimension to each projector $P \in \A,$ and then being extended by continuity to all of $\A_+.$

Putting these considerations together, we have the following definition:
\begin{definition} \label{def:renormalized-trace}
	A trace $\tau : \A_+ \to [0, \infty]$ will be called a \textbf{renormalized trace} if it satisfies the following three properties.
	\begin{enumerate}[(i)]
		\item If $\tau(\rho)$ vanishes, then we have $\rho=0.$ (This means that the trace is ``faithful.'')
		\item If $P$ is an algebraically finite projector (cf. definition \ref{def:projector-equivalence}), then we have $\tau(P) < \infty.$ (I have decided to call this property ``cleverness,'' i.e., a trace is ``clever'' if it knows that finite projectors should have finite trace. This is related to a notion called ``semifiniteness''; see remark \ref{rem:clever-trace}.)
		\item If $\{\rho_{\alpha}\}$ is a family of positive operators in $\A_+$ with a supremum $\rho = \sup_{\alpha} \rho_{\alpha},$ then we have $\tau(\rho) = \sup_{\alpha} \tau(\rho_{\alpha}).$ (This means that the trace is ``normal.'')
	\end{enumerate}
\end{definition}
\begin{remark} \label{rem:clever-trace}
	The idea of calling a renormalized trace ``clever'' is something I made up for these notes.
	Instead of considering a trace that is faithful, normal, and clever, mathematicians would generally consider a trace that is faithful, normal, and \textit{semifinite}.
	Semifiniteness means that for every nonzero $T \in \A$, there exists some nonzero $S\in \A$ with $S \leq T$ and $\tau(S) < \infty.$
	In a factor of type I or II, cleverness and semifiniteness are the same, because there is some spectral projector $P$ of $T$ and some positive number $\lambda$ such that we have $T \geq \lambda P,$ and $P$ must dominate some finite projector $F$, so if we take $S = \lambda F$ then we have $T \geq S$ and $\tau(S) < \infty.$
	I have chosen to introduce the notion of ``cleverness'' because it will allow us to say that every type III factor has a unique faithful, normal, clever trace, even though it has no faithful, normal, semifinite trace.
\end{remark}
We may now appeal to a beautiful theorem in the theory of von Neumann algebras which tells us that every von Neumann factor $\A$ has a unique renormalized trace up to rescaling.
I.e., if $\tau$ and $\tau'$ are two renormalized traces on the factor $\A$, then there exists some positive number $\lambda > 0$ with $\tau = \lambda \tau'.$
I will sketch the proof of this theorem in subsection \ref{subsec:trace-construction}; for now, let us explore its consequences.

First and foremost, the rule for assigning expectation values given by equation \eqref{eq:expectation-rule} is independent of the overall normalization of $\tau,$ so it is an intrinsic property of the von Neumann algebra that does not depend on how we scale the renormalized trace.
This makes it a real physical feature of $\A$!

Next, we note that the trace of two projectors that are equivalent in the sense of definition \ref{def:projector-equivalence} have the same trace.
This follows from the comment I made after definition \ref{def:trace}, which implies that that $\tau$ satisfies $\tau(VV^*) = \tau(V^* V)$ for every $V \in \A.$
Consequently, a renormalized trace assigns the value $+\infty$ to every infinite projector.
This is because if $P$ is an infinite projector in $\A$ (cf. definition \ref{def:projector-equivalence}), there exists a nontrivial subprojector $Q \in \A$ with $Q \lneq P$ and $Q \sim P$.
We then have
\begin{equation}
	\tau(P) = \tau(Q) + \tau(P - Q) = \tau(P) + \tau(P - Q).
\end{equation}
So if $\tau(P)$ is finite then we have $\tau(P - Q) = 0,$ which implies $P - Q = 0,$ hence $P = Q,$ which is a contradiction.

Using similar reasoning, we may also argue that if $P$ and $Q$ are finite with $\tau(P) = \tau(Q)$, then we have $P \sim Q.$
To show this, we use the comparison theorem (theorem \ref{thm:comparison}), which tells us that we have either $P \precsim Q$ or $Q \precsim P$; assuming without loss of generality $P \precsim Q,$ i.e. $P \sim R \leq Q,$ we obtain
\begin{equation}
	\tau(Q) = \tau(Q - R) + \tau(R) = \tau(Q - R) + \tau(P) = \tau(Q - R) + \tau(Q),
\end{equation}
hence $\tau(Q - R) = 0,$ which implies $Q = R$ and therefore $P \sim Q.$

We have observed, therefore, that the effective dimension assigned to a projector by $\tau$ is a useful invariant of its equivalence class.
If $P$ and $Q$ are finite, the statement $P \sim Q$ is equivalent to $\tau(P) = \tau(Q)$, and if $P$ is infinite, then we have $\tau(P) = \infty.$
If we appeal to proposition \ref{prop:infinite-projectors-equivalent}, which says that in a factor any two infinite projectors are equivalent, we may conclude that even in the infinite case, we have $\tau(P) = \tau(Q)$ if and only if $P \sim Q.$
These observations are summarized in the following theorem.
\begin{theorem}
	Let $\tau$ be a renormalized trace on the factor $\A$.
	For two projectors $P, Q \in \A$, we have $\tau(P) = \tau(Q)$ if and only if we have $P \sim Q.$
	In fact, $\tau$ is order preserving, i.e., we have $\tau(P) \leq \tau(Q)$ if and only if we have $P \precsim Q.$
	Consequently, $\tau$ is an order isomorphism from the ordered set of projector equivalence classes in $\A$ to some subset of $[0, \infty].$ 
\end{theorem}

It is interesting to ask which subset of $[0, \infty]$ is obtained by acting on projectors in $\A$ with $\tau.$
In the type III case, every nonzero projector is infinite by definition, so the values assigned to projectors by $\tau$ are $\{0, \infty\}.$
In the type I case, one can show that any two nonzero minimal projectors are equivalent --- this follows immediately from the comparison theorem \ref{thm:comparison} --- and furthermore that any projector $P$ can be written as a sum of pairwise orthogonal nonzero minimal projectors --- this can be proven by Zorn's lemma, since if a maximal family of pairwise orthogonal nonzero minimal projectors $Q_j$ with $Q_j \leq P$ fails to satisfy $\sum_j Q_j = P,$ then $P - \sum_j Q_j$ has a nonzero minimal subprojector, which contradicts maximality of the family $Q_j$.
So the allowed values of the trace are determined entirely by the value that the trace assigns to a nonzero minimal projector.
If we call this value $\lambda,$ then the allowed values are $\{0, \lambda, 2 \lambda, \dots, n \lambda\},$ where $n$ can be finite or $\infty$ and denotes the maximal number of pairwise orthogonal minimal projectors in the algebra.
In the type II case, the trace assigns a value $\lambda$ to the identity operator which is finite in the II$_{1}$ case and infinite in the type II$_{\infty}$ case, and every possible real number between $0$ and $\lambda$ is an allowed value for the trace.
This is because for any $\lambda_0 \in [0, \lambda],$ we can set\footnote{Recall that in fact \ref{fact:projector-infima-and-suprema} we showed that $P_-$ and $P_+$ are projectors in $\A$.}
\begin{align}
	P_-
		& = \sup \{P \in \A \text{ a projector with } \tau(P) \leq \lambda_0\}, \\
	P_+
		& = \inf\{P \in \A \text{ a projector with } \tau(P) \geq \lambda_0\}.
\end{align}
If $P_+$ is equivalent to $P_-,$ then we have $\lambda_0 \leq \tau(P_+) = \tau(P_-) \leq \lambda_0,$ so $\tau(P_+) = \lambda_0$ and $\lambda_0$ is an allowed value for $\tau.$
If $P_+$ is not equivalent to $P_-,$ then by the comparison theorem (theorem \ref{thm:comparison}), we must have $P_+ \succnsim P_-,$ i.e. $P_+ \sim Q \gneq P_-.$
Thus $Q - P_-$ is a nonzero projector, and since it cannot be minimal (as we are in a type II factor), there must be a nonzero projector $R \in \A$ with $R \lneq Q - P_-$.
But then we have $\lambda_0 < \tau(P_- + R) < \tau(Q) = \tau(P_+),$ which contradicts the definition of $P_+.$

We summarize the results of the above paragraph in the following theorem.
\begin{theorem} \label{thm:projector-trace-values}
	In a type I$_n$ factor with renormalized trace $\tau,$ the possible values of the trace on projectors are $\{0, \lambda, 2 \lambda, \dots, n \lambda\}$, where $\lambda$ is an arbitrary constant related to the normalization of $\tau$.
	If the factor is type $I_{\infty},$ the possible values of the trace on projectors are $\{0, \lambda, \dots, \infty\}.$
	By convention, one often takes $\lambda = 1.$
	
	In a type II factor with renormalized trace $\tau,$ the possible values of the trace on projectors are $[0, \lambda],$ where $\lambda$ is infinity in the type II$_\infty$ case, or an arbitrary finite value related to the normalization of $\tau$ in the type II$_1$ case.
	Conventionally, in the type II$_1$  case, one takes $\lambda = 1.$
	
	In a type III factor with renormalized trace $\tau,$ the possible values of the trace on projectors are $\{0, \infty\}.$
\end{theorem}

This classification helps us understand the heuristic given at the end of section \ref{sec:big-picture} for the type classification of factors in terms of the kinds of renormalizable states that it contains, and the relation of that heuristic to the formal definition of types given in section \ref{subsec:algebraic-mixed-states}.
If a factor is type III, then the renormalized trace is the trivial one that assigns $0$ to the zero operator and $+\infty$ to every nonzero positive operator.
Consequently, a type III factor has no renormalizable density operators.
If a factor is type I, then the renormalized trace assigns finite values to at least some projectors, so it has at least a few renormalizable density operators.
Furthermore, there exist minimal projectors, which are renormalizable, and which are pure in the sense that they cannot be expressed as positive linear combinations of other renormalizable density operators in the algebra.
Finally, in a factor of type II, there exist projection operators that are assigned finite trace by $\tau$ and hence are renormalizable, but every one of these can be expressed as linear combinations of smaller projectors that also correspond to renormalizable states.

Another useful fact to note is that the Hilbert space trace ``$\tr$'' is always faithful and normal, but not necessarily clever.
If the factor $\A$ contains a single positive operator whose Hilbert space trace is finite, then ``$\tr$'' is itself a renormalized trace, and every renormalized trace is proportional to the ordinary Hilbert space trace ``$\tr.$''
So things only get really interesting when every positive operator in $\A$ has infinite Hilbert space trace, since in this case the renormalized trace is different from ``$\tr.$''

As a final comment, I will describe an interesting fact about the density matrices associated to a renormalized trace.
The claim is that in a factor $\A$ of type I or II with renormalized trace $\tau$, a large class of algebraic states on $\A$ can be represented by formal limits of renormalizable density matrices.
In algebraic quantum theory, a state on $\A$ is a linear functional $\omega : \A \to \comps$ that satisfies $\omega(T^* T) \geq 0$ and $\omega(1) = 1.$
This state is said to be ``normal'' if it is continuous with respect to a particular topology called the ultraweak topology.\footnote{The ultraweak topology is not discussed in these notes, but see for example section 20 of \cite{conway2000course}.}
One can show that for any normal state $\omega : \A \to \comps,$ there exists a \textit{formal density operator $\rho_{\omega}$} satisfying $\omega(T) = \tau(\rho_{\omega} T)$ for all $T \in \A.$
The reason I say ``formal density operator'' instead of just ``density operator'' is because $\rho_{\omega}$ might not actually be an operator in $\A.$
The actual statement is that there exists some sequence $\rho_{\omega, n}$ of operators in $\A_1,$ which is Cauchy with respect to the norm $\lVert x \rVert_1 = \tau(|x|),$ such that we have $\omega(T) = \lim_{n} \tau(\rho_{\omega,n} T)$ for all $T \in \A$.
For further details, see theorem 2.18 and proposition 2.19 of chapter V in \cite{takesaki2001theory}.

\subsection{Construction of the unique trace on a factor}
\label{subsec:trace-construction}

In the preceding subsection, I claimed that every von Neumann factor admits a unique renormalized trace (see definition \ref{def:renormalized-trace}) up to rescaling.
In this subsection, I will sketch a proof of that claim.
Though I will not explain all of the details, I will give references for the details as needed.
The idea of the proof is to show that once the value of $\tau$ on a single nonzero finite projector has been chosen, its values on all other projectors are completely determined, and once its values on projectors are determined its value on any positive operator can be obtained via continuity.

First note that in any factor, the functional $\tau_{\infty} : \A_+ \to [0, \infty]$ that sends the zero operator to zero, and every nonzero positive operator to $+\infty$, satisfies every property of definition \ref{def:renormalized-trace} except for cleverness.
Cleverness requires that any algebraically finite projector has finite renormalized trace.
This is vacuously satisfied by $\tau_{\infty}$ in the case of a type III factor, since type III factors have no algebraically finite projectors!
So in the type III case, $\tau_{\infty}$ is a renormalized trace.
To see that it is the only one, let $\tau$ be a generic renormalized trace for a type III factor $\A$.
As explained in theorem \ref{thm:projector-trace-values}, every nonzero projector $P$ must satisfy $\tau(P) = \infty$.
If $\rho \in \A_+$ is a nonzero positive operator, then it has at least one nonzero spectral projection $P$ for which there is a positive number $\alpha$ satisfying $\rho \geq \alpha P.$
The basic properties of the trace given in definition \ref{def:trace} then guarantee $\tau(\rho) \geq \alpha \tau(P) = \infty,$ so in a type III factor any renormalized trace must assign $+\infty$ to every nonzero positive operator.
It is therefore equal to $\tau_{\infty}.$

The next simplest case to check is that of a type I factor.
A type I factor has nonzero minimal projectors, and as discussed in the paragraph preceding theorem \ref{thm:projector-trace-values}, all minimal projectors in a type I factor are equivalent.
So to construct a renormalized trace on $\tau,$ we may start by choosing a positive number $\lambda$ to assign as the renormalized trace of each nonzero minimal projector.
That is, if $P$ is a nonzero minimal projector, we will define $\tau(P) = \lambda.$
I also explained in the paragraph preceding theorem \ref{thm:projector-trace-values} that every projector $Q \in \A$ can be expressed as a sum of pairwise orthogonal minimal projectors; since the trace $\tau$ is supposed to be additive, we must assign $\tau(Q) = n \lambda,$ where $n$ is the number (possible infinity) of minimal projectors in the decomposition of $Q$.
These considerations uniquely determine the value of $\tau$ on every projector up to the choice of $\lambda,$ i.e., up to an overall scaling.
It is a general fact (see appendix \ref{app:spectral-theory}) that every positive operator $\rho$ is the supremum over all positive linear combinations of projections that satisfy
\begin{equation}
	\sum_{j=1}^{n} p_j P_j \leq \rho.
\end{equation}
If $\tau$ is to be additive, then it must satisfy
\begin{equation}
	\tau \left( \sum_{j=1}^n p_j P_j \right)
		= \sum_{j=1}^n p_j \tau(P_j),
\end{equation}
and if it is to be normal (see definition \ref{def:renormalized-trace}), then it must satisfy
\begin{equation}
	\tau(\rho)
		= \sup \left\{ \tau \left( \sum_{j} p_j P_j \right) | \sum_j p_j P_j \leq \rho \right\}
		= \sup \left\{ \sum_{j} p_j \tau(P_j) | \sum_j p_j P_j \leq \rho \right\}
\end{equation}
So once the value of $\tau$ is determined on projections, its value on every positive operator is determined by the additivity condition in definition \ref{def:trace} and the normality condition in definition \ref{def:renormalized-trace}.
Again, the only freedom we had in defining $\tau$ was to define its value on a minimal projector, which extends to an overall rescaling ambiguity for $\tau : \A_+ \to [0, \infty].$

So far, we have shown that if a renormalized trace $\tau$ on a type I factor exists, then it must be unique up to rescaling.
We have explained how to determine the value of $\tau$ on any projector after choosing what its value is on a minimal projector. and then how to use this data to determine the value of $\tau$ on any positive operator.
We have \textit{not} shown, however, that the function $\tau$ thus defined satisfies all the properties of the trace --- in particular, it is not guaranteed that $\tau$ is additive on the space of positive operators.
The basic issue is that if we have $\sum_{j} p_j P_j = \sum_{j} q_j Q_j$ for two distinct positive linear combinations of projectors, we have not yet established the identity $\sum_{j} p_j \tau(P_j) = \sum_j q_j \tau(Q_j).$
Showing this is very tricky, and I will not actually do it in this note.
I will have more to say about the proof of existence, but I will say it after explaining uniqueness of the renormalized trace on a type II factor, since the existence issues for type I and type II factors are equivalent.

Putting aside the type I case momentarily, let us discuss the case of a type II factor, which is significantly more involved.
I will sketch the proof in a heuristic way, but technical details can be found from proposition 1.3.5 to proposition 1.3.11 in \cite{sunder2012invitation}.
To start, note that a type II factor has at least one nonzero finite projector; let us fix one, and call it $P.$
We will construct a renormalized trace $\tau$ by picking an arbitrary positive value for $\tau(P)$ --- say, $\tau(P) = \lambda.$
We now want to show that if $\tau$ is to satisfy the conditions of a renormalized trace, then its value on any other projector is determined by its value on $P.$
The key tool here will be a version of the remainder theorem for division of natural numbers, which says that for any positive integers $p, q,$ there is an integer $a$ and an integer $r$ with $0 \leq r < q$, satisfying
\begin{equation}
	p = a q + r.
\end{equation}
We will analogously show that in a type II factor, for any finite nonzero projectors $P$ and $Q,$ the projector $P$ can be expressed as a finite sum over pairwise-orthogonal projectors $\{Q_j\}$ equivalent to $Q$, plus a ``remainder'' $R$ satisfying $R \precnsim Q$ and $R$ orthogonal to each $Q_j.$
I find it helpful to think of this in terms of a picture, sketched in figure \ref{fig:remainder-theorem}.
Without loss of generality, we may assume $Q \lneq P.$
Then I will represent $P$ as a box, roughly indicating a subspace of $\H$, and represent $Q$ as a smaller box, indicating the fact that $Q$ has cardinality strictly less than that of $P.$
We can think of the division algorithm for projectors as saying: ``put together as many copies of the box $Q$ as you can within the box $P$, then call whatever is left over $R$, and the comparison theorem implies $R \precnsim Q.$''

\begin{figure}[h!]
	\centering
	\includegraphics{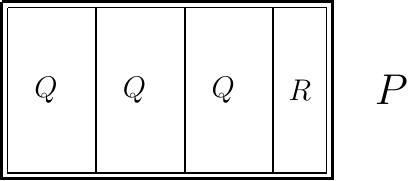}
	\caption{A sketch of the remainder theorem for finite projectors in a factor $\A$. For any finite $P \in \A$ and any $Q \in \A$ with $Q \precsim P,$ there exists a family $Q_1, \dots, Q_n$ of pairwise-orthogonal projectors equivalent to $Q$, and some projector $R \precnsim Q$ satisfying $R \perp Q_j,$ with $P = Q_1 + \dots + Q_n + R.$}
	\label{fig:remainder-theorem}
\end{figure}

In mathematical terms, the remainder theorem is proven by using induction to construct a maximal family $\{Q_j\}$ of subprojectors of $P$ that are pairwise orthogonal and equivalent to $Q$.
Maximality of this family, together with the comparison theorem (theorem \ref{thm:comparison}), then implies that the remainder $R = P - \sum_j Q_j$ satisfies $R \precnsim Q.$
We therefore have
\begin{equation}
	P = \sum_j Q_j + R.
\end{equation}
Furthermore, the total cardinality of the set $\{Q_j\}$ must be finite; for if it were countably infinite, the relations $Q_{n+1} \sim Q_n$ together with the additivity lemma (lemma \ref{lem:additivity-lemma}) would imply
\begin{equation}
	P = \sum_{j=1}^{\infty} Q_j + R \sim \sum_{j=2}^{\infty} Q_j + R \lneq P,
\end{equation}
contradicting the finiteness of $P$.
It is also not hard to show --- see for example proposition 1.3.5 of \cite{sunder2012invitation} --- that the cardinality of $\{Q_j\}$ is an intrinsic property of the projectors $P$ and $Q$ that is independent of the specific decomposition $P = \sum_{j=1}^{n} Q_j + R.$

The basic idea is now to show, using the fact that our factor is type II and thus contains no minimal projectors, that for every positive integer $n$ there exists a projector $F_n$ such that $P$ can be represented as the sum of $n$ mutually orthogonal projectors equivalent to $F_n,$ with no remainder.
One sets $F_n$ equal to the supremum over all subprojectors $S \leq P$ such that the division of $P$ by $S$ includes at least $n$ copies of $S$, then shows that this supremum divides $P$ exactly $n$ times with no remainder.
I will not show this explicitly, as the proof is not particularly enlightening, but it follows from the considerations in chapter 1.3 of \cite{sunder2012invitation} (though regrettably this result is not stated explicitly there).
If $\tau$ is to be additive, it must then satisfy $\tau(F_n) = \lambda / n.$
The value of $\tau(Q)$ for any finite projector $Q$ can then be determined by ``squeezing it'' between multiples of the projectors $F_n.$
For each $F_n,$ we can divide $Q$ by $F_n$ and write $Q$ as an integer number $m$ of copies of $F_n,$ plus a remainder.
The trace $\tau,$ if it exists, must satisfy $\tau(Q) \geq m \tau(F_n) = m \lambda / n,$ and $\tau(Q) < (m+1) \tau(F_n) = (m+1) \lambda / n.$
It is then straightforward to show --- again, consult \cite{sunder2012invitation} for details --- that as $n$ increases, the numbers $m \lambda / n$ and $(m+1) \lambda / n$ get arbitrarily close to one another, which forces $\tau(Q)$ to take on a specific value.

We have thus argued that on a type II factor, any renormalized trace has its value on any finite projector determined by its value on a single finite projector that is used as a reference.
To any infinite projector, we must assign trace infinity.
As in the type I case, the value of $\tau$ on a generic positive operator $\rho$ is then completely determined by the spectral theorem together with the additivity of the trace (see definition \ref{def:trace}) and its normality (see definition \ref{def:renormalized-trace}).
We have therefore established uniqueness of the renormalized trace on a type II factor up to rescaling, but we are left with the same existence issues as in the type I case.

The proof of existence amounts to showing that the function $\tau$, which is defined on $\A_+$ by defining its action on projectors and extending that action using the spectral theorem, is well defined and satisfies all the properties of definitions \ref{def:trace} and \ref{def:renormalized-trace}.
This turns out to be highly nontrivial.
A treatment of this issue in the case of a finite factor was given by Murray and von Neumann in \cite{murray1937rings}.
A much shorter treatment was given by Yeadon in \cite{yeadon1971new}, though the shortness of the proof is only possible by application to a technical result called the Ryll-Nardzewski fixed point theorem.
A complete treatment of the existence of the trace for arbitrary factors (in fact, for arbitrary von Neumann algebras) can be found in chapter V.2 of \cite{takesaki2001theory}.
This last treatment starts with Yeadon's proof of the trace for finite factors, and obtains the general case by an appropriate decomposition of an infinite algebra into finite pieces.
I do not recommend trying to understand these proofs in detail; it took me quite a long time to read them, and I do not feel that my understanding of the subject has been significantly enhanced by that adventure.

\section{The standard forms of certain factors}
\label{sec:standard-forms}

In this section, I show how to find the ``standard forms'' of the type I and type II$_{\infty}$ factors --- a type I factor being isomorphic to an algebra of the form $\B(\H) \otimes 1_{\H'},$ and a type II$_{\infty}$ factor being isomorphic to an algebra of the form $\B(\H) \otimes \A_{\H'}$ with $\A_{\H'}$ a type II$_1$ factor.
Comparing these two cases is conceptually helpful, as it demonstrates that the type I and II$_{\infty}$ cases are structurally similar --- both can be thought of as the algebra $\B(\H)$ together with some ``internal algebra,'' the difference being that in the type I case the internal algebra is trivial while in the type II$_{\infty}$ case the internal algebra is of type II$_1.$

\subsection{The standard form of a type I factor}

Consider first the case where we have a type I factor $\A$.
By definition, $\A$ has a minimal nonzero projector, which we will call $P$.
Using induction, we can construct a maximal family of pairwise-orthogonal projectors $\{P_j\},$ each of which is equivalent to $P$.
Maximality of this family, together with the comparison theorem (theorem \ref{thm:comparison}), implies that we have $\sum_{j} P_j = 1,$ for otherwise we would have $1 - \sum_j P_j \precnsim P,$ which would contradict the minimality of $P$.
I sketch this structure heuristically in figure \ref{fig:type-I-structure}, where I represent the algebra as a (possibly infinite) row of boxes of equal size.
Each box represents a projector equivalent to $P$, and all of the projectors are pairwise orthogonal.
The boxes are empty, to represent the fact that the projectors are minimal; the algebra has no structure at a level more granular than a single box.
By definition, if $n$ (possibly $n=\infty$) is the number of boxes in this figure, then the factor is of type I$_n.$

\begin{figure}[h!]
	\centering
	\includegraphics{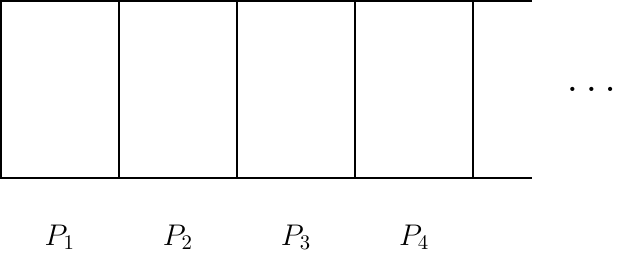}
	\caption{A heuristic drawing of the structure of a type I$_n$ factor. There is some set $P_1, P_2, \dots$ of pairwise-orthogonal minimal projectors that sum up to the identity.
	The total number of projectors in this family (possibly $\infty$) is $n.$}
	\label{fig:type-I-structure}
\end{figure}

If the minimal projector $P$ has rank one, then the projectors $\{P_j\}$ correspond to an orthonormal basis of $\H$, and $\A$ is the full algebra $\B(\H).$
What if $P$ has rank greater than one?
Well, \textit{morally}, \textit{algebraically}, $P$ has \textit{effective} rank one.
Can we make this precise?
Yes!
The idea is to construct a new Hilbert space, $\H',$ which has an orthonormal basis labeled by the projectors $P_j.$
Formally, we define a vector $\ket{P_j}$ for each projector $P_j,$ define the inner product so that these vectors are orthonormal, and take $\H'$ to be the space
\begin{equation}
	\H' = \left\{ \sum_{j} c_j \ket{P_j} | \sum_{j} |c_j|^2 < \infty\right\}.
\end{equation}
Relative to the original algebra, I like to think of the vectors in this Hilbert space as wavefunctions over the space of boxes in figure \ref{fig:type-I-structure}.
The leftover structure is just whatever is contained in a single box, which is like the trivial algebra made up of the identity operator acting on $P \H.$

To make all this formal, we note that since each $P_j$ is equivalent to $P$, there exist partial isometries $V_{j}$ mapping $P_j \H$ to $P \H.$
If we pick some orthonormal basis $\{\ket{e_\mu}\}$ for $P \H,$ then we obtain an orthonormal basis for each $P_j \H$ as $\{V_j^* \ket{e_{\mu}} \}.$
The collection $\{V_{j}^* \ket{e_{\mu}}\}$ for all $j$ and $\mu$ is then an orthonormal basis for $\H$, and the collection $\{\ket{P_j} \otimes \ket{e_\mu} \}$ is an orthonormal basis for $\H' \otimes P \H$, so we can obtain a unitary map $U$ from $\H$ to $\H' \otimes P\H$ via the identification
\begin{equation}
	U \left( V_{j}^* \ket{e_{\mu}} \right)
		= \ket{P_j} \otimes \ket{e_\mu}.
\end{equation}
If we conjugate our original algebra $\A$ by $U$, we obtain an algebra $U \A U^*$ that acts on $\H' \otimes P \H$.
We can understand this algebra by computing the matrix elements
\begin{align}
	(\bra{P_j} \otimes \bra{e_\mu}) U T U^* (\ket{P_k} \otimes \ket{e_\nu})
		& = \bra{V_j^* e_{\mu}} T \ket{V_k^* e_{\nu}} \nonumber \\
		& = \bra{e_\mu} V_j T V_{k}^* \ket{e_\nu}.
\end{align}
The operator $V_j T V_{k}^*$ is an operator in $\A$ that acts only within the subspace $P \H$.
Consequently, it must be a scalar multiple of $P.$
If it were not, then we could split it into its hermitian and antihermitian parts, and find that the spectral projections of these parts --- which are in $\A$ by fact \ref{fact:spectral-theorem} --- include a proper nonzero subprojection of $P$, contradicting the assumption that $P$ is minimal.
We may therefore write $V_j T V_{k}^* = T_{jk} P$, and conclude
\begin{align}
	(\bra{P_j} \otimes \bra{e_\mu}) U T U^* (\ket{P_k} \otimes \ket{e_\nu})
	& = T_{jk} \delta_{\mu \nu}.
\end{align}
From this formula, we find that $U \A U^*$ is contained within the algebra $\B(\H') \otimes 1_{P \H}.$
The reverse inclusion is easy, since for any operator $S$ in $\B(\H') \otimes 1_{P \H},$ we can compute its matrix elements $S_{jk} \delta_{\mu \nu}$ and specify an operator $T$ in $\B(\H)$ via the formula $V_j T V_{k}^* = S_{jk},$ which guarantees that we have $U T U^* = S.$

So via unitary conjugation, any type I$_n$ factor can be mapped to the algebra $\B(\H') \otimes 1_{P \H},$ where $\H'$ is a Hilbert space of dimension $n$, and $1_{P \H}$ keeps track of the trivial internal structure of the minimal projectors.

\subsection{The standard form of a type II$_{\infty}$ factor}

A type II$_{\infty}$ factor $\A$ has no minimal projector, but it does contain at least one finite projector $P$.
As in the preceding subsection, we can decompose the identity operator as a sum $1 = \sum_j P_j,$ where each $P_j$ is equivalent to $P$ and the $\{P_j\}$ projectors are pairwise orthogonal.
In analogy with figure \ref{fig:type-I-structure}, I will represent the II$_{\infty}$ algebra in figure \ref{fig:type-II-structure} as a row of boxes, each one corresponding to one of the projectors $P_j.$
Since the factor is infinite, there are infinitely many boxes.
While in figure \ref{fig:type-I-structure} the boxes were empty, to indicate that there was no structure to the algebra within each box, in figure \ref{fig:type-II-structure} the boxes are ruled to indicate that there is further granularity to the algebra within each projector $P_j.$
In fact, we will now see that each box contains a type II$_1$ factor.

\begin{figure}[h!]
	\centering
	\includegraphics{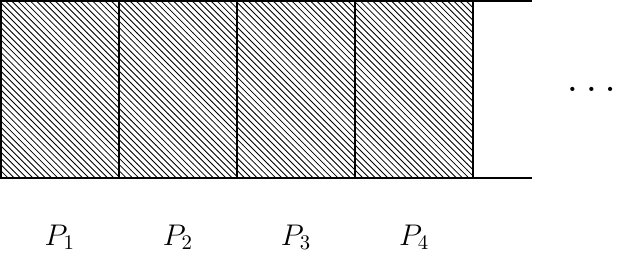}
	\caption{A heuristic drawing of the structure of a type II$_\infty$ factor. There is some infinite sequence $P_1, P_2, \dots$ of pairwise-orthogonal finite projectors that sum up to the identity.
	These projectors are not minimal; they have internal structure.
	In particular, each algebra $P_j \A P_j,$ considered as an algebra acting on $P_j \H,$ is a type II$_1$ factor.}
	\label{fig:type-II-structure}
\end{figure}

This is because for any projector $P \in \A$, the algebra $P \A P$, considered as an algebra acting on $P \H,$ is a von Neumann factor.\footnote{We already used this fact previously in section \ref{subsec:factor-decomposition}. It  can be found as proposition 43.8 of \cite{conway2000course}.}
The algebra $P \A P$ contains no minimal projector, since any minimal projector in $P \A P$ would also be minimal in $\A$.
Furthermore, its identity operator is $P$, which is a finite projector.
So $P \A P,$ considered as a von Neumann algebra on $P\H,$ is a type II$_1$ factor.

Proceeding exactly as in the previous subsection, we can construct a Hilbert space $\H'$ whose orthonormal basis is labeled by the projectors $P_j.$
There is a family of partial isometries $V_j$ mapping $P_j \H$ onto $P \H$, and if $\{\ket{e_\mu}\}$ is a basis for $P\H$ then $\{V_j^* \ket{e_{\mu}}\}$ is a basis for $\H$.
We construct the unitary identification from $\H$ to $\H' \otimes P \H$ given by $U : V_j^* \ket{e_\mu} \mapsto \ket{P_j} \otimes \ket{e_\mu}.$

Following exactly the same logic as in the preceding section, one can show that the algebra $U \A U^*$ acts on $\H' \otimes P\H$ as
\begin{equation}
	U \A U^* = \B(\H) \otimes P \A P.
\end{equation}
So $\A$ is unitarily equivalent to a tensor product of an algebra of bounded operators $\B(\H),$ which I think of as representing how $\A$ is spread about among the different boxes in figure \ref{fig:type-II-structure}, together with a type II$_1$ factor.

Note that this decomposition is somewhat less canonical than the one given in the preceding subsection.
This is because while there is a unique equivalence class of minimal projectors in a type I factor, there are many inequivalent finite projectors in a type II$_{\infty}.$
Nevertheless, I find this decomposition useful for understanding structural similarities and differences between the type I and type II cases.
Both are represented by algebras of bounded operators on a Hilbert space, but type I has a trivial internal structure while type II$_\infty$ has the internal structure of a type II$_1$ factor.

\section{Miscellany}
\label{sec:miscellany}

\subsection{Modular theory}
\label{sec:modular-theory}

This subsection is intended only for readers who are already familiar with modular theory, so I will not explain its fundamentals.
A nice explanation of the basics of modular theory can be found in \cite{witten2018aps}.
The point is that given a von Neumann algebra $\A$ with a cyclic, separating vector $\ket{\Omega},$ one can construct an unbounded, self-adjoint, positive operator $\Delta_{\Omega}$ called the ``modular operator'' satisfying two important properties:
\begin{enumerate}[(i)]
	\item For any $t \in \reals,$ we have $\Delta_{\Omega}^{it} \A \Delta_{\Omega}^{-it} = \A.$
	\item We have $\Delta_{\Omega} \ket{\Omega} = \ket{\Omega}.$
\end{enumerate}
These conditions tell us that the flow $\Delta^{it}$ is a symmetry for the expectation values of $\A$ in the state $\ket{\Omega}.$

There are some interesting relationships between the unbounded operator $\Delta_{\Omega}$ and the type classification of von Neumann algebras.
Namely, there is an algebraic property of $\Delta_{\Omega}$ that can be used to tell whether or not a factor is type III (though it does not distinguish between type I and II).
It is also possible to use algebraic properties of $\Delta_{\Omega}$ to study the different kinds of type III factors; we will return to this point in the next subsection.
The essential theorem is the following:
\begin{theorem}
	Let $\A$ be a factor, and $\ket{\Omega}$ a cyclic and separating vector with modular operator $\Delta_{\Omega}.$
	If the spectrum of $\Delta_{\Omega}$ does not include zero, then $\A$ is not type III.
\end{theorem}
\begin{proof}[Sketch of proof]
	The proof of this theorem is rather complicated; I will only sketch its major steps, then provide a reference for further reading.
	
	Suppose we have a modular operator $\Delta_{\Omega}$ whose spectrum does not include zero.
	The idea is to use $\Delta_{\Omega}$ to construct a renormalized trace on $\A$ (see definition \ref{def:renormalized-trace}) that assigns finite trace to at least one bounded operator in $\A$.
	By the discussion in section \ref{sec:traces}, this will imply that $\A$ is not type III.
	
	Since $\Delta_{\Omega}$ does not contain zero in its spectrum, the operator $\Delta_{\Omega}^{-1}$ exists and is bounded.
	It is a general fact in modular theory that any modular operator is related to its inverse by antiunitary conjugation; consequently, since $\Delta_{\Omega}^{-1}$ is bounded, $\Delta_{\Omega}$ must be bounded as well.
	One can then appeal to a technical result --- see e.g. corollary 4.1.14 of \cite{sakai2012c} --- which states that there must exist some unitary flow $U(t) \in \A$ with $U(t) T U(t)^* = \Delta_{\Omega}^{it} T \Delta_{\Omega}^{-it}.$\footnote{This is nontrivial because $\Delta_{\Omega}^{it}$ is not in $\A$; the result says that the modular flow can be produced using unitaries that lie within $\A$ itself.}
	One then shows that the flow $V(t) = U(t)^* \Delta^{it}$ is also unitary, and lies in $\A'.$
	One constructs a Hamiltonian $H'$ generating this flow, i.e., satisfying $V(t) = e^{i t H'}.$
	There is then a special set of vectors in the domain of $\exp(H'/2)$ such that one can define the functional
	\begin{equation}
		\tau(\ketbra{x}{y}) = \braket{e^{H'/2} y}{e^{H'/2} x}
	\end{equation}
	for any $\ket{x}, \ket{y}$ in this set.
	The set of vectors on which this functional is defined is sufficiently dense in $\H$ for $\tau$ to be defined on all of $\A_+$, and one must then show that the extended trace satisfies all the conditions of definition \ref{def:renormalized-trace}.
	Details can be found in section 14 of \cite{takesaki2006tomita}.
\end{proof}

It is important to note that a single von Neumann algebra could have multiple cyclic and separating vectors.
So what the above theorem is really saying is, if a factor is type III, then the modular operator associated with \textit{any} cyclic and separating state must have zero in its spectrum.
This observation leads us to two natural questions:
\begin{enumerate}[(i)]
	\item Is the converse true?
	That is, if every modular operator associated to a cyclic and separating state has zero in its spectrum, then is the factor necessarily type III?
	\item If so, then is there a way to determine whether a factor is type III using only a single modular operator, rather than checking the spectra of every possible modular operator?
\end{enumerate}

The answer to the first question is \textit{sort of}.
We can give the answer ``yes'' provided that we enlarge our definition of what we mean by a modular operator.
While we have discussed so far modular flows associated with cyclic and separating vectors in $\H$, there is a more general notion of a modular flow associated with a \textit{faithful semifinite weight}.
I will describe this briefly; more information can be found in chapter 2 of \cite{sunder2012invitation}.

A faithful, semifinite weight is sort of like an algebraic state, except that it need not assign finite expectation values to all operators in the algebra.
Formally, it is a map $\phi : \A_+ \to [0, \infty]$ that satisfies conditions (i) and (ii) of definition \ref{def:trace}, such that $\phi(T) = 0$ iff $T = 0,$ and such that the set $\A_{\phi, +}$ of elements with $\phi(T) < \infty$ is appropriately dense in $\A_+$; see section 2.4 of \cite{sunder2012invitation} for more details.
Similarly to the case of a trace, one can define an ideal $\A_{\phi}$ of Hilbert-Schmidt operators (cf. appendix \ref{app:trace-math}).
One then defines an inner product on the Hilbert-Schmidt operators by $\braket{X}{Y} = \phi(X^* Y),$ and turns $\A_{\phi}$ into a Hilbert space using the GNS construction.
It turns out that one can define a modular operator $\Delta_{\phi}$ acting on this Hilbert space associated with the weight $\phi,$ and since $\A_{\phi}$ carries a natural action of $\A$ via multiplication on the left, conjugation by $\Delta_{\phi}^{it}$ can be used to produce a flow on $\A$.

The idea is that if $\A$ is type I or II, then its renormalized trace \textit{is} a faithful semifinite weight, and the modular operator associated with it is just the identity.\footnote{The modular operator $\Delta_{\Omega}$ is defined so that for $a, b \in \A$ we have
\begin{equation}
	\bra{a \Omega} \Delta_{\Omega} \ket{b\Omega} = \braket{b^* \Omega}{a^* \Omega} = \bra{\Omega} b a^* \ket{\Omega}.
\end{equation}
If $\ket{\Omega}$ is obtained by performing the GNS construction on a trace, then by cyclicity of the trace we have
\begin{equation}
	\bra{a \Omega} \Delta_{\Omega} \ket{b\Omega} = \bra{\Omega} b a^* \ket{\Omega} = \bra{\Omega} a^* b \ket{\Omega} = \braket{a \Omega}{b \Omega}.
\end{equation}
Since $\ket{\Omega}$ is cyclic for $\A$, vectors of the form $a \ket{\Omega}$ are dense in the appropriate Hilbert space, so $\Delta_{\Omega}$ must be the identity operator.
}
So if $\A$ is type I or II, then in our enlarged notion of what a modular operator can be, there necessarily exists a modular operator for which zero is not in the spectrum, namely, the modular operator associated with the trace.
So we may say with confidence: ``A factor is type III if and only if every modular operator associated to a faithful semifinite weight has zero in its spectrum.''

We now return to the second question stated above.
Can the question of whether a factor is type III be determined from a single modular operator, without needing to check all of them?
The answer is yes, but the construction is quite complicated, and I will not try to describe it here.
The idea is that there is a subset of the nonnegative reals associated with any von Neumann algebra $\A$, often denoted $S(\A)$, which consists of the intersection of the spectra of all possible modular operators for $\A$.
This set was studied by Connes in \cite{connes1973classification}, where he gave a way to reconstruct this set using information about how the modular flow of a single modular operator acts on projectors in $\A$.
Details of the correspondence are reviewed in chapters 3.2 and 3.3 of \cite{sunder2012invitation}.
Knowing the set $S(\A)$ is certainly sufficient to know whether there exists a modular operator without zero in its spectrum, so this data can be used to determine whether a factor is of type III.

\subsection{Type III$_1$ factors in quantum field theory}
\label{sec:type-III1}

In the previous subsection, I mentioned the set $S(\A)$ associated with a von Neumann algebra $\A$, which is defined as the intersection over all spectra of all possible modular operators associated to $\A$.
In \cite{connes1973classification}, Connes defined a factor to be of type III$_1$ if $S(\A)$ is the full set $[0, \infty),$ that is, if every modular operator associated to $\A$ has all the nonnegative reals in its spectrum.
It is commonly believed that in quantum field theory, the von Neumann factors associated with subregions should generally be of type III$_1.$
Why is this the case?

The starting point of the argument is the observation due to Bisognano and Wichmann \cite{bisognano1975duality, bisognano1976duality} that if $\ket{\Omega}$ is the vacuum state of a quantum field theory in Minkowski spacetime, and $\A$ is the von Neumann algebra of observables associated to a Rindler wedge, then the modular flow $\Delta_{\Omega}^{it}$ corresponds to the Killing boosts that are null on the Rindler horizon.
In this case, $\Delta_{\Omega}$ clearly has spectrum $[0, \infty)$.
If $\ket{\Psi}$ is another cyclic and separating state in the vacuum sector of the theory, then the modular flow generated by $\Delta_{\Psi}$ is not expected to be expressible in a simple closed form.
But because every state in a quantum field theory is supposed to ``look like the vacuum at short distances,'' the modular flow generated by $\Delta_{\Psi}$ is expected to look \textit{approximately} like a boost near the bifurcation surface of the Rindler wedge.
From this heuristic, one tries to argue that $\Delta_{\Psi}$ should also have spectrum covering all of $[0, \infty),$ and therefore to conclude that the von Neumann algebra associated with the Rindler wedge in a general quantum field theory is of type III$_1.$
Since this argument is very local in the sense that it really only cares about physics in a small neighborhood of the bifurcation surface of the Rindler wedge, it is expected to apply to more general regions, since every smooth domain of dependence looks ``approximately Rindler'' near its edge.

While this intuition is useful, it is insufficient for a few reasons.
First, of course, all the hand-wavy notions of ``approximately Rindler'' and ``approximately like a boost'' need to be made precise for us to be sure that they are meaningful; we are firmly in the realm of speculation, and a little more math is needed to be sure that we are actually saying something that makes sense.
Second, we have only discussed modular flows associated to Hilbert space states in the vacuum sector, but the set $S(\A)$ is defined in terms of the modular flows associated with \textit{all} faithful semifinite weights.
A great deal of work has been done to argue that every modular flow associated to a quantum field theory algebra should have $[0, \infty)$ as its spectrum.
The unifying idea is to come up with a precise definition of what it means for a theory to ``look conformally invariant in the UV,'' and then to show that in any algebraic quantum field theory with this property, any modular operator can be compared with the vacuum modular operator in such a way that establishes its spectrum as being $[0, \infty).$
Broad reviews of the literature on this topic can be found in \cite{yngvason2005role} and in chapter V.6 of \cite{haag2012local}.
Relevant research papers include \cite{araki1964type, driessler1977type, fredenhagen1985modular, buchholz1987universal, buchholz1995scaling}.

\subsection{Non-factor algebras}
\label{sec:non-factor}

What can one say about the classification of von Neumann algebras that are not factors?
The simplest thing is to appeal to section \ref{subsec:factor-decomposition}, where I explained how any von Neumann algebra can be decomposed in terms of factors.
One can then say that a von Neumann algebra $\A$ is type I (respectively II or III) if every factor appearing in this decomposition is type I (respectively II or III).
Obviously, then, not every von Neumann algebra has a definite type.
For example, the direct sum of a factor of type I and a factor of type II is neither type I nor type II!

There is a more algebraic way of defining the type classification of general von Neumann algebras, which mimics the classification given in terms of projectors in section \ref{sec:projector-types}.
There, we said that a factor is of type I if it has a nonzero minimal projector, of type II if it has a nonzero finite projector but no nonzero minimal projectors, and of type III if it has no nonzero finite projectors.
We can frame this in a way that suggests generalization by saying:
\begin{definition}
\leavevmode \vspace{-\baselineskip}
\begin{enumerate}[(i)]
	\item A factor is of type I if there is a nonzero minimal projector dominated by the identity.
	\item A factor is of type II if there is a nonzero finite projector dominated by the identity, but no nonzero minimal projector dominated by the identity.
	\item A factor is of type III if there is no nonzero finite projector dominated by the identity.
\end{enumerate}
\end{definition}
When $\A$ is a factor, the identity operator is the only nonzero projection in the center of $\A$.
When $\A$ has a nontrivial center, we might be tempted to generalize the type classification as follows:
\begin{definition}[Tentative, and incorrect] \label{def:bad-definition}
	\leavevmode \vspace{-\baselineskip}
\begin{enumerate}[(i)]
	\item \textit{(Not correct.)} A von Neumann algebra is of type I if every nonzero central projection dominates some nonzero minimal projection.
	\item \textit{(Not correct.)} A von Neumann algebra is of type II if it has no nonzero minimal projections, but every nonzero central projection dominates some nonzero finite projection.
	\item A von Neumann algebra is of type III if it has no nonzero finite projections.
\end{enumerate}
\end{definition}
We should observe immediately that there is some asymmetry in these definitions that was not present in the factor case.
Namely, an algebra can be somewhere ``between types I and II'' under these definitions if it does have \textit{some} nonzero minimal projections, but not enough nonzero minimal projections for every central projection to dominate one.
This is actually a feature, not a bug; it is what allows for us to have von Neumann algebras which do not fit into any type, as described in the first paragraph of this subsection.
There is a problem with these definitions, however, which is that the notion of a ``minimal projection'' stops being the right one to consider once we generalize to non-factor algebras.
That is to say, if we try to define the types of general von Neumann algebras using definition \ref{def:bad-definition}, it will not coincide with the natural definition that a von Neumann algebra has a given type if and only if its decomposition into factors is made up entirely of factors of that type.
Instead of a \textit{minimal projector}, one must instead consider an \textit{abelian projector}, which is a projector $P \in \A$ such that the algebra $P \A P$ is abelian.
It is not hard to show that if $\A$ is a factor, then $P$ is abelian if and only if it is minimal.
To define types that match with the factor decomposition of an algebra, it is necessary to use the following definition.
\begin{definition}
	\leavevmode \vspace{-\baselineskip}
	\begin{enumerate}[(i)]
		\item A von Neumann algebra is of type I if every nonzero central projection dominates some nonzero abelian projection.
		\item A von Neumann algebra is of type II if it has no nonzero abelian projections, but every nonzero central projection dominates some nonzero finite projection.
		\item A von Neumann algebra is of type III if it has no nonzero finite projections.
	\end{enumerate}
\end{definition}
For details, the reader is encouraged to consult chapter V.1 of \cite{takesaki2001theory}.

In this more general setting, the issue of the renormalized trace (see definition \ref{def:renormalized-trace}) is more subtle.
The definition of a renormalized trace as being ``faithful, clever, and normal'' is no longer appropriate.
One simply gives up on trying to define a trace on von Neumann algebras that include type III subfactors, which liberates us from considering the ``cleverness'' property of a trace (see remark \ref{rem:clever-trace}).
If a von Neumann algebra has no type III element in its factor decomposition, then it is called semifinite, and one can show that it has a trace that is faithful, normal, and semifinite (see again remark \ref{rem:clever-trace}).
This trace is no longer unique up to rescaling.
The general theorem, given as theorem 2.31 in \cite{takesaki2001theory}, is that if $\tau_1$ and $\tau_2$ are two faithful, normal, semifinite traces, then there exists some central element $z \in \Z$ with $0 \leq z \leq 1$ satisfying
\begin{equation}
	\tau_1(T) = (\tau_1 + \tau_2)(z T)
\end{equation}
and
\begin{equation}
	\tau_2(T) = (\tau_1 + \tau_2)((1-z) T).
\end{equation}
It is not hard to see that when $\A$ is a factor, and $z$ is just a scalar multiple of the identity, this reduces to the claim made in section \ref{sec:traces} that the trace is unique up to rescaling.

\acknowledgments{It is a pleasure to thank Daniel Harlow for helpful conversations about central operators in gauge theories, and Adam Levine for letting me test some of this article's pedagogy on him.
I also thank Barton Zwiebach and Christian Ascione for comments on an earlier version of this manuscript.
I especially thank Netta Engelhardt for encouraging me to use the first few months of my postdoc to learn something new.
My research is supported by the Templeton Foundation via the Black Hole Initiative.}

\appendix

\section{Some mathematical background}
\label{app:math-background}

The rigorous study of von Neumann algebras can be extremely technical.
This is a consequence of the fact that while many of the results are algebraic or topological in nature, the techniques needed to prove them are analytic.
To deal with this issue, I have tried throughout these notes to give precise statements and references without including all the gory mathematical detail of the proofs.
While most of the technical material is not necessarily illuminating, I have identified two subjects that I think are really worthwhile to learn if you want to understand how to think properly about von Neumann algebras.
These topics are \textit{nets and topology} and \textit{the algebraic spectral theorem}.
This appendix explains those subjects at a mathematical level.

\subsection{Nets and topology}

Von Neumann algebras are algebras of operators that are closed in a particular topology --- \textbf{the weak operator topology}.
It is essential to understand this topology in order to understand not only how von Neumann algebras work, but \textit{why they are defined the way they are.}
The standard topologies on bounded operators can be formulated very easily by specifying the conditions under which a net of operators converges --- ``when does $T_{\alpha}$ converge to $T$?''
A net is a generalization of a sequence, and it has the nice property that any topology on any space is completely specified by its convergent nets; this is not true for sequences, as distinct topologies can have the same convergent sequences.
We will discuss nets in section \ref{sec:nets} in order to describe the standard topologies on operators in section \ref{app:operator-topologies}.

\subsubsection{Nets}

\label{sec:nets}

A sequence in a space $X$ is a map $x : \mathbb{N} \to X,$ where $\mathbb{N} = \{1, 2, \dots\}$ is the set of natural numbers.
We write points in the image of this map as $x_n \equiv x(n).$
The natural numbers come with an ordering, $1 < 2 < \dots,$ and we think of this ordering as being inherited by the points $\{x_n\},$ so that we say ``$\{x_n\}$ is eventually contained within the set $S \subseteq X$ if there is some integer $N$ such that $n \geq N$ implies $x_n \in S.$''
A \textbf{net} is a set of points in $X$ that has an ordering that may be different from the ordering of the natural numbers, but for which the notion of the set being ``eventually contained in $S$'' still makes sense.

The starting point is a \textbf{directed set} $\Lambda,$ which is a set with a special kind of partial order.
For any two points $\lambda_1, \lambda_2 \in \Lambda,$ we might have $\lambda_1 \leq \lambda_2,$ or we might have $\lambda_2 \leq \lambda_1,$ or we might not be able to compare the elements at all -- this is what makes the order \textit{partial}.
A typical example of a partial order is a tree, where we say two vertices $v_1$ and $v_2$ are comparable if they lie in the same branch of the tree, and we write $v_2 \leq v_1$ if $v_2$ lies further from the root than $v_1.$ See figure \ref{fig:tree-directed-set}.

\begin{figure}[h!]
	\centering
	\includegraphics{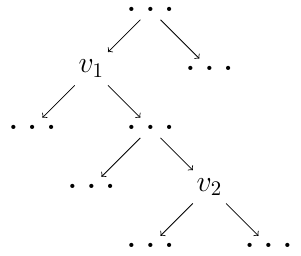}
	\caption{Any tree induces a partial order on its vertices, where we say $v_2 \leq v_1$ if $v_2$ lies further from the root than $v_1$.
	This partial order has the property that for any two vertices $v$ and $v',$ there exists some third vertex $v*$ that is closer to the root than either of them.
	Any partially ordered set with this property is called a \textit{directed set}.}
	\label{fig:tree-directed-set}
\end{figure}

If the set $\Lambda$ has a partial order that can be sketched as a tree, then it has a special property that is not true of a generic partial order: for any $\lambda_1, \lambda_2 \in \Lambda,$ there exists some third element $\lambda_* \in \Lambda$ that is greater than both of them: $\lambda_1 \leq \lambda_*, \lambda_2 \leq \lambda_*.$
See again figure \ref{fig:tree-directed-set}.
A \textbf{directed set} is a partially ordered set with this property.
A \textbf{net} in a space $X$ is a map $x : \Lambda \to X,$ with points in the image denoted by $x_{\lambda} \equiv x(\lambda).$
We can think of the points $\{x_{\lambda}\}$ as inheriting the order on $\Lambda,$ so we say ``$\{x_\lambda\}$ is eventually contained within the set $S \subseteq X$ if there is some $\lambda_0 \in \Lambda$ such that $\lambda \geq \lambda_0$ implies $x_{\lambda} \in S.$''
The ``directed set'' property of $\Lambda$ is important for making sense of this definition, because it guarantees that for any $\lambda \in \Lambda,$ there exists some $\tilde{\lambda}$ with $\tilde{\lambda} \geq \lambda_0, \tilde{\lambda} \geq \lambda,$ so every $x_{\lambda}$ is followed in the order by some $x_{\tilde{\lambda}}$ that is in $S$, and furthermore every point in the net following $x_{\tilde{\lambda}}$ is in $S$.

When $X$ is a topological space, we say \textbf{the net $x_{\lambda}$ converges to the point $x \in X$} if for every open set $U$ containing $x$, $\{x_{\lambda}\}$ is eventually contained within $U$.

Nets have two very useful properties, which I will now state.
For proofs, I encourage you to consult chapter 2 of \cite{kelley2017general}.
\begin{enumerate}[(i)]
	\item A topology is completely determined by its convergent nets.
	That is, if $X$ is a space and $\mathcal{T}_1, \mathcal{T}_2$ are two topologies on that space such that we have $x_{\lambda} \to x$ in $\mathcal{T}_1$ if and only if $x_{\lambda} \to x$ in $\mathcal{T}_2,$ then we have $\mathcal{T}_1 = \mathcal{T}_2.$
	
	This is a consequence of the more general fact that we have\footnote{Remember that a topology is a list of open sets; $\mathcal{T}_1 \subseteq \mathcal{T}_2$ means that every set which is open with respect to $\mathcal{T}_1$ is open with respect to $\mathcal{T}_2.$ In this case, we say that $\mathcal{T}_1$ is \textbf{weaker} or \textbf{coarser} than $\mathcal{T}_2.$ We say $\mathcal{T}_2$ is \textbf{stronger} or \textbf{finer} than $\mathcal{T}_1.$} $\mathcal{T}_1 \subseteq \mathcal{T}_2$ if and only if every net in $X$ that converges with respect to $\mathcal{T}_2$ converges with respect to $\mathcal{T}_1.$
	The slogan is: ``In a weaker topology, more nets converge.''
	
	\item Continuity of a function is completely determined by nets.
	That is, if $X$ and $Y$ are topological spaces, and $f : X \to Y$ is a function, then $f$ is continuous if and only if for every convergent net $x_{\lambda} \to x$ in $X$, the net $f(x_{\lambda})$ converges to $f(x)$ in $Y$.
\end{enumerate}

\subsubsection{Operator topologies}
\label{app:operator-topologies}

Now we are ready to discuss the standard topologies on $\B(\H)$, the space of bounded operators on a Hilbert space $\H$.
The space of bounded operators is a vector space; any reasonable topology should respect this vector space structure, meaning that operator addition $+ : \B(\H) \times \B(\H) \to \B(\H)$ and scalar multiplication $\cdot : \comps \times \B(\H) \to \B(\H)$ should be continuous maps.
For a general vector space $V$, there is a very general kind of topology with this property called a \textbf{locally convex} or \textbf{seminorm} topology, which is induced by a family of distance measures.
We start with some family $\{p_{\alpha}\}$ of \textbf{seminorms} on $V$.
These are maps $p_{\alpha} : V \to [0, \infty)$ that satisfy the triangle inequality and $p_{\alpha}(c \ket{v}) = |c| p_{\alpha}(\ket{v}),$ together with $p_{\alpha}(0) = 0$ and $p_{\alpha}(\ket{x}) \geq 0,$ but are not necessarily positive definite in that they may assign $p_{\alpha}(\ket{v}) = 0$ for $\ket{v} \neq 0.$
Distance measures with these properties show up all the time in mathematics and physics, and given some seminorm family $\{p_{\alpha}\},$ it is useful to put a topology on $V$ for which the notion of closeness is induced by the seminorm distances.
It turns out that there does exist such a topology, generated by fundamental neighborhoods that are intersections of finitely many balls of the form
\begin{equation}
	B_{\alpha, \ket{x}}(r) = \{ \ket{y} \in V \text{ such that } p_{\alpha}(\ket{y} - \ket{x}) < r\}.
\end{equation}
The seminorm topology has two nice properties:
\begin{itemize}
	\item It is the weakest (i.e., coarsest) topology in which addition is continuous, scalar multiplication is continuous, and all of the seminorms $p_{\alpha}$ are continuous.
	\item The topology is completely characterized by the fact that the net $\ket{x_{\lambda}}$ converges to $\ket{x}$ if and only if $p_{\alpha}(\ket{x_{\lambda}} - \ket{x})$ converges to zero for every seminorm $p_{\alpha}.$
\end{itemize}
For more details on seminorm topologies and proofs of all claims made above, I suggest consulting chapter 1 of \cite{rudin1991functional}.

Now, for the particular vector space $\B(\H)$, there are three topologies that are of interest to us in these notes: the norm topology, the strong topology, and the weak topology.

The \textbf{norm topology} is a seminorm topology induced by the operator norm
\begin{equation}
	\lVert T \rVert = \sup_{\lVert \ket{x} \rVert = 1} \lVert T \ket{x} \rVert.
\end{equation}
So \textit{in the norm topology, we have $T_{\alpha} \to T$ if and only if $\lVert T_{\alpha} - T \rVert \to 0.$} For this reason, the norm topology is sometimes called ``the topology of uniform convergence.''

The \textbf{strong topology} or \textbf{strong operator topology} is a seminorm topology induced by a family of seminorms that are labeled by the vectors $\ket{x} \in \H$:
\begin{equation}
	p_{\ket{x}}(T) = \lVert T \ket{x} \rVert.
\end{equation}
So \textit{in the strong topology, we have $T_{\alpha} \to T$ if and only if $\lVert (T - T_{\alpha}) \ket{x} \rVert \to 0$ for every $\ket{x} \in \H$, or equivalently if and only if $T_{\alpha} \ket{x} \to T \ket{x}$ for every $\ket{x} \in \H$.} For this reason, the strong topology is sometimes called ``the topology of pointwise convergence.''

The \textbf{weak topology} or \textbf{weak operator topology} is a seminorm topology induced by a family of seminorms that are labeled by pairs of vectors $\ket{x}, \ket{y} \in \H$:
\begin{equation}
	p_{\ket{x}, \ket{y}}(T) = | \bra{x} T \ket{y}|.
\end{equation}
So \textit{in the weak topology, we have $T_{\alpha} \to T$ if and only if $\bra{x}T_{\alpha} \ket{y} \to \bra{x} T \ket{y}$ for every $\ket{x} \in \H$.} For this reason, the weak topology is sometimes called ``the topology of matrix element convergence.''

One can show fairly easy that if $T_{\alpha} \to T$ in the norm topology, then we have $T_{\alpha} \to T$ in the strong topology.
If $T_{\alpha} \to T$ in the strong topology, then we have $T_{\alpha} \to T$ in the weak topology.
This gives the ordering
\begin{equation}
	\text{norm} \supseteq \text{strong} \supseteq \text{weak}.
\end{equation}
A useful consequence of this ordering is that if a set is closed in the weak operator topology, then it is closed in the strong operator topology.
In the proof of fact \ref{fact:commutant-vN}, it was mentioned that every von Neumann algebra is weakly closed.
Consequently, it is strongly closed; this fact is used in section \ref{sec:relative-dimension}.

\subsection{The algebraic spectral theorem}
\label{app:spectral-theory}

There are many ways of thinking about the spectral theorem for bounded operators.
I will give a brief explanation of the way presented in \cite{douglas1998banach}.
The essential statement of the spectral theorem is given in the main text as fact \ref{fact:spectral-theorem}; the purpose of this appendix is to lay out the essential propositions that must be proved to lead to the spectral theorem, and indicate where proofs can be found.

\begin{definition}
	Let $T \in \B(\H)$ be a bounded operator.
	The \textbf{spectrum} of $T$, denoted $\sigma(T)$, is the set of complex numbers $\lambda \in \comps$ for which $T - \lambda$ (that is, $T$ minus $\lambda$ times the identity operator) is not invertible.  
\end{definition}

\begin{prop}
	If $T$ is Hermitian, then its spectrum is contained in $\reals.$
	If $T$ is positive, then its spectrum is contained in $[0, \infty).$
\end{prop}
\begin{proof}[Sketch of proof]
	This amounts to showing that if $T$ is Hermitian and $\lambda$ has nonzero imaginary part, then $T - \lambda$ is invertible, then showing that if $T$ is positive and $\lambda$ is negative, then $T - \lambda$ is invertible.
	The proofs use a lemma that an operator $S$ is invertible if both $S$ and $S^*$ are bounded below.
	See propositions 4.8 and 4.15 of \cite{douglas1998banach}.
\end{proof}

\begin{lemma} \label{lem:neumann-series}
	Any operator $S$ with $\lVert 1 - S \rVert < 1$ is invertible.
\end{lemma}
\begin{proof}[Sketch of proof]
	One proves this lemma by constructing an inverse as the geometric series
	\begin{equation}
		S^{-1} = \sum_{n=0}^{\infty} (1 - S)^{n},
	\end{equation}
	and checking that the series converges in the norm topology to an inverse of $S$.
	The details are worked out in proposition 2.5 of \cite{douglas1998banach}.
\end{proof}

\begin{prop}
	The spectrum of $T$ is contained in the closed ball around $0 \in \comps$ of radius $\lVert T \rVert.$
\end{prop}
\begin{proof}
	Suppose that $\lambda \in \comps$ has $|\lambda| > \lVert T \rVert.$
	Then we want to show that $T - \lambda$ is invertible.
	To do this, we use the observation
	\begin{equation}
		\left\lVert 1 - \left( 1 - \frac{T}{\lambda} \right) \right\rVert < 1,
	\end{equation}
	which implies by lemma \ref{lem:neumann-series} that $1 - T/\lambda$ is invertible.
	Hence $\lambda - T$ is invertible, and so is $T - \lambda.$
\end{proof}

\begin{prop}
	The spectrum of $T$ is compact and nonempty.
\end{prop}
\begin{proof}[Sketch of proof]
	In light of the preceding proposition, which shows that $\sigma(T)$ is bounded, showing that it is compact requires only showing that it is topologically closed.
	To do this, we show that its complement --- the set of all $\lambda \in \comps$ such that $T - \lambda$ is invertible --- is open.
	This set is the preimage of the space $\operatorname{GL}(\H)$ of invertible elements in $\B(\H)$ under the map $\lambda \mapsto T - \lambda.$
	So we can prove compactness of $\sigma(T)$ by proving that $\operatorname{GL}(\H)$ is open in the norm topology, and proving that the map $\lambda \mapsto T - \lambda$ is norm-continuous.
	Neither of these is hard to do; continuity of the map is straightforward, and openness of $\operatorname{GL}(\H)$ follows from simple manipulations starting with lemma \ref{lem:neumann-series}.
	For details, consult proposition 2.7 of \cite{douglas1998banach}.
	
	The proof that $\sigma(T)$ is nonempty is very similar to the complex-analytic proof that every polynomial has a root.
	The idea is that if $\sigma(T)$ is empty, then one can show that the map $\lambda \mapsto \lVert (T - \lambda)^{-1}\rVert $ is a bounded, entire function of the complex plane.
	Appealing to Liouville's theorem then tells us that it must be constant, and since it limits to zero for large $|\lambda|$, this implies $\lVert (T - \lambda)^{-1} \rVert = 0$ and hence $(T - \lambda)^{-1} = 0$ for all $\lambda,$ which is a contradiction since zero is not the inverse of any operator.
	Some care must be taken in filling out the details of this proof to establish that the map is bounded and entire; for details, see theorem 2.29 of \cite{douglas1998banach}. 
\end{proof}

\begin{definition}
	A \textbf{$C^*$-subalgebra} of $\B(\H)$ is a $*$-subalgebra with identity (see definition \ref{def:*-subalgebra}) $\C \subseteq \B(\H)$ that is closed in the norm topology (see appendix \ref{app:operator-topologies}).
\end{definition}

\begin{remark}
	Note that every von Neumann algebra is also a $C^*$ algebra, since von Neumann algebras are weakly closed (see the proof of fact \ref{fact:commutant-vN}) and weakly closed sets are norm-closed (see section \ref{app:operator-topologies}).
\end{remark}

\begin{theorem} [$C^*$ spectral theorem] \label{thm:C*-spectral}
	Let $T$ be a normal operator in $\B(\H),$ and let $\C(T)$ be the minimal $C^*$ algebra containing $T$ (i.e., the norm-closure of all polynomials in $T$ and $T^*$).
	Let $C(\sigma(T))$ be the space of continuous functions from $\sigma(T)$ to $\comps$.
	Then the map
	\begin{align}
		\Gamma : & \, \C(T) \to C(\sigma(T)) \nonumber \\
						& \, \sum_{j, k} c_{j k} T^j (T^*)^k \mapsto \sum_{j, k} c_{j k} z^j \bar{z}^k
	\end{align}
	is an isomorphism, and takes the operator norm on $\C(T)$ to the supremum norm $\lVert f \rVert = \sup_{z \in \sigma(T)} |f(z)|$ on $C(\sigma(T)).$
\end{theorem}
\begin{proof}[Sketch of proof]
	The proof of this theorem is very nontrivial, and I will only sketch its major steps.
	First, one constructs the space $\M_{\C(T)}$ of nontrivial homomorphisms $\phi : \C(T) \to \comps,$ and endows this space with a topology such that we have $\phi_{\alpha} \to \phi$ if and only if we have $\phi_{\alpha}(S) \to \phi(S)$ for all $S \in \C(T).$
	One then constructs a map $\Gamma$ from $\C(T)$ to continuous functions on $\M_{\C(T)}$ via
	\begin{align}
		[\Gamma(S)](\phi) = \phi(S).
	\end{align}
	A brilliant lemma is used to show that $\M_{\C(T)}$ is in bijection with the space of \textit{maximal ideals} in $\C(T)$; this step uses normality of $T$, since the map is only a bijection if $\C(T)$ is abelian.
	The identification of $\M_{\C(T)}$ with the maximal ideals of $\C(T)$ is then used to show that the map $\Gamma$ preserves invertibility; that is, $S \in \C(T)$ is invertible if and only if the map $\Gamma(S)$ is nonzero on all of $\M_{\C(T)}.$
	Since invertibility is what is used to determine whether or not a certain number $\lambda$ is in the spectrum of an operator, this lets us ask questions about spectra in $\C(T)$ by asking about the values attained by certain continuous functions on $\M_{\C(T)}.$
	By studying the consequences of this identification for polynomials in $T$ and $T^*$, and extending those observations to all of $\C(T)$ using the Stone-Weierstrass theorem, one can show that $\Gamma$ is an algebraic isomorphism, and furthermore that it takes the operator norm on $\C(T)$ to the supremum norm:
	\begin{equation}
		\sup_{\phi \in \M_{\C(T)}} |[\Gamma(S)](\phi)| = \lVert S \rVert.
	\end{equation}
	The final step is to show that $\M_{\C(T)}$ is homeomorphic to $\sigma(T)$ under the map $\phi \mapsto \phi(T)$, which completes the proof.
	
	As is apparent from the sketch given above, many lemmas and propositions must be stacked atop one another to complete the proof of this theorem.
	The final theorem can be found as theorem 4.30 of \cite{douglas1998banach}, with appeals to material from chapter 2 of that same book.
\end{proof}

\begin{remark}
	Why should the theorem we have just stated be thought of as a version of the spectral theorem?
	In finite dimensions, the real reason we \textit{want} a spectral theorem is to be able to apply functions to operators.
	If an operator $T$ is positive, we \textit{define} $\sqrt{T}$ by diagonalizing $T$ and taking the square root along its diagonal.
	The theorem above gives us a way to do this so long as the function we want to apply is continuous on the spectrum of the operator in question.
	It tells us that for any normal operator $T \in \B(\H)$, and any continuous function $f : \sigma(T) \to \comps,$ there exists some operator $f(T) \in \C(T)$ that has all the same algebraic properties as the function $f$!
	This is extremely useful.
	It can be used, for example, to show that every positive operator $T$ has a unique positive square root, and furthermore that the square root is contained in the $C^*$ algebra generated by $T$.
	For more details, see e.g. proposition 4.33 of \cite{douglas1998banach}.
	
	We are not done yet, though; we often want to find operators corresponding to discontinuous functions on the spectrum of $T$.
	For example, if $\omega \subseteq \sigma(T)$ is an open set, and $\chi_{\omega} : \sigma(T) \to \{0, 1\}$ is the indicator function for that set, we expect to be able to define the operator $\chi_{\omega}(T)$ as a ``projection onto that subset of the spectrum.''
	Theorem \ref{thm:C*-spectral} doesn't give us operators of this form, and in fact \textit{it forbids those operators from being contained in $\C(T)$}, since the map $\Gamma : \C(T) \to C(\sigma(T))$ is an isomorphism.
	
	What we really want is to enlarge the algebra $\C(T)$ so that it contains operators corresponding to \textit{bounded} functions of $\sigma(T)$, not just continuous functions.
	It will turn out that the right enlargement is to take the von Neumann algebra generated by $T$, which is generally larger than $\C(T)$ and is in fact isomorphic to the space of bounded functions on $\sigma(T)$.
\end{remark}

\begin{background}[Basic Lebesgue theory]
	This definition lays out the basic tools of Lebesgue theory.
	Consult chapter 1 of \cite{rudin1974real} for more detail.
	
	A space $X$ is said to be \textbf{measurable} if it is equipped with a family $\Sigma$ of subsets of $X$ that contains $X$, is closed under complements, and is closed under countable unions.
	The sets in $\Sigma$ are called the \textbf{measurable sets}.
	
	A map $f : X \to Y$ where $X$ is measurable and $Y$ is a topological space is said to be a \textbf{measurable function} if the preimages of open sets are measurable.
	
	A \textbf{measure} $\mu$ on a space $X$ is a map $\mu : \Sigma \to [0, \infty]$ that assigns $\mu(\varnothing) = 0$ for the empty set, and that is countably additive on disjoint measurable sets.
	That is, if $\{E_j\}$ is a sequence of pairwise-disjoint measurable sets, then we have $\mu(\cup_j E_j) = \sum_j \mu(E_j).$
	
	A \textbf{simple function} on $X$ is a measurable function $s : X \to [0, \infty]$ that takes only finitely many  values in its image.
	I.e., it is a positive linear combination of indicator functions for measurable sets.
	
	The \textbf{Lebesgue integral} of a simple function $s = \sum_{j} p_j \chi_{\omega_j}$ is defined by
	\begin{equation}
		\int_{X} d\mu\, s
			\equiv \sum_j p_j \mu(\omega_j).
	\end{equation}

	The \textbf{Lebesgue integral} of a measurable function $f : X \to [0, \infty]$ is the supremum over the integrals of all simple functions it dominates:
	\begin{equation}
		\int_{X} d\mu\, f
			\equiv \sup\left\{ \int_{X} d\mu\, s \text{ such that } s \leq f \right\}.
	\end{equation}
	
	The space $L^2(X, \mu)$ is the set of all functions $f : X \to \comps$ for which $|f|^2$ has finite Lebesgue integral, quotiented by all functions that vanish away from a set of measure zero.
	It is a Hilbert space with respect to the inner product
	\begin{equation}
		\braket{f}{g}
			\equiv \int_{X} d\mu\, \bar{f} g.
	\end{equation}

	The space $L^{\infty}(X, \mu)$ is the set of all bounded functions $f : X \to \comps$, quotiented by all functions that vanish away from a set of measure zero.
	It is a Banach space with respect to the norm
	\begin{equation}
		\lVert f \rVert_{\infty}
			= \inf\{\text{$\alpha \geq 0$ such that $|f(x)| > \alpha$ occurs only on a set of measure zero}\}.
	\end{equation}

	For every $\phi \in L^{\infty}(X, \mu),$ there is a \textbf{multiplication operator} $M_\phi$ acting on $L^2(X, \mu)$ as $M_{\phi} \ket{f} = \ket{\phi f}.$
	The operator norm of $M_{\phi}$ is the $\infty$-norm of $\phi,$ i.e., we have $\lVert M_{\phi} \rVert = \lVert \phi \rVert_{\infty}.$
	The space of multiplication operators is an abelian von Neumann subalgebra of $\B(L^2(X, \mu)).$
\end{background}

\begin{theorem}[Full spectral theorem]
	Let $\H$ be a separable Hilbert space, and $T \in \B(\H)$ a normal operator.
	Let $\C(T)$ be the $C^*$ algebra generated by $T$ and $\A(T)$ the von Neumann algebra generated by $T$.
	Then there exists a measure $\mu$ on the spectrum $\sigma(T)$ and a map $\hat{\Gamma} : \A(T) \to \B(L^2(\sigma(T), \mu))$ satisfying the following conditions:
	\begin{enumerate}[(i)]
		\item the image of $\hat{\Gamma}$ is exactly the space $L^{\infty}(\sigma(T), \mu)$ of multiplication operators on $L^2(\sigma(T), \mu).$
		$\hat{\Gamma}$ is an isomorphism with respect to the algebraic structure and the weak operator topology on both spaces.
		\item The measure $\mu$ is a Borel measure, meaning the set of measurable subsets is generated by the open subsets of $\sigma(T).$
		This implies that continuous functions from $\sigma(T)$ to $\comps$ are measurable.
		\item Since continuous functions on a compact set are bounded, this implies that every continuous function on $\sigma(T)$ is contained in $L^{\infty}(\sigma(T), \mu).$ The map $\hat{\Gamma}$ extends the one given in theorem \ref{thm:C*-spectral} in that for any continuous $f \in L^{\infty}(\sigma(T), \mu),$ we have $f(T) = \hat{\Gamma}^{-1}(M_f).$
		
		\item Every open subset of $\sigma(T)$ has nonzero measure.
	\end{enumerate}
	Consequently, for any bounded Borel-measurable function $f : \sigma(T) \to \comps,$ we can \textbf{define} $f(T)$ by $\hat{\Gamma}^{-1}(M_f).$
\end{theorem}

\begin{proof}[Sketch of proof]
	The proof of this theorem is even more complicated than the proof of theorem \ref{thm:C*-spectral}.
	The workhorse tool is a fundamental theorem in measure theory called the \textit{Riesz-Markov theorem},\footnote{See e.g. chapter 2 of \cite{rudin1974real}} which says that for any linear functional $\Lambda : C(\sigma(T)) \to \comps,$ there exists a measure $\mu$ on $\sigma(T)$ satisfying $\Lambda(f) = \int_{\sigma(T)} d\mu\, f$.
	This means that for any linear functional on $C(\sigma(T)),$ there exists a Borel measure on $\sigma(T)$ that makes this linear functional equivalent to integration.
	
	The basic idea is to construct a vector $\ket{\Omega} \in \H$ that is \textit{cyclic} for $\C(T)$, meaning that the closure of $\C(T) \ket{\Omega}$ is the full Hilbert space $\H$.
	One then constructs a linear functional $\Lambda : C(\sigma(T)) \to \comps$ by
	\begin{equation}
		\Lambda(f) = \bra{\Omega} f(T) \ket{\Omega},
	\end{equation}
	and finds a corresponding Riesz-Markov measure $\mu$ on $\sigma(T)$.
	Using theorem \ref{thm:C*-spectral}, you can then construct a map from $\C(T) \ket{\Omega}$ to $C(\sigma(T))$ by
	\begin{equation}
		U (f(T) \ket{\Omega}) = f,
	\end{equation}
	and show that this is an isometry with respect to the $L^2$ norm on $C(\sigma(T))$ induced by the measure $\mu$.
	Taking closures in both the domain and the range, one produces a unitary operator $U$ from $\bar{\C(T) \ket{\Omega}}$ to $L^2(\sigma(T), \mu),$ and sets $\hat{\Gamma}(S) = U S U^*.$
	
	Various topological arguments --- using lemmas that $\C(T)$ is weakly closed in $\A(T)$ and $C(\sigma(T))$ is weakly closed in $L^{\infty}(\sigma(T), \mu)$ --- then prove the theorem.
	The main subtlety not addressed here is that $\A(T)$ is not guaranteed to have a cyclic vector.
	To get around this, one finds a vector $\ket{\Omega}$ with sufficiently nice properties that the above argument can be applied to the subspace $\bar{\C(T) \ket{\Omega}}$ and then extended to the full Hilbert space. d
	
	The details of this proof are very involved, and can be found in chapter 4 of \cite{douglas1998banach}.
\end{proof}

\begin{remark}
	As a simple corollary, let us prove a statement used in section \ref{subsec:trace-construction}, which states that every positive operator $\rho$ is the supremum over all positive linear combinations of projections that satisfy
	\begin{equation}
		\sum_j p_j P_j \leq \rho.
	\end{equation}
	We certainly have inequality, that is,
	\begin{equation}
		\rho \geq \sup\left\{ \sum_j p_j P_j \text{ such that } \sum_j p_j P_j \leq \rho \right\}.
	\end{equation}
	Equality follows from the measure theory fact that every $L^{\infty}$ function can be written as the supremum over all simple functions that it dominates.
	Using the spectral theorem, $\rho$ can be treated as an $L^{\infty}$ function, and the simple functions that it dominates are positive linear combinations of its spectral projectors.
	This completes the proof.
\end{remark}

\section{Properties of traces}
\label{app:trace-math}

Recall the definition of a trace on a von Neumann algebra given in definition \ref{def:trace}.
\begin{definition} \label{def:trace-appendix}
	A \textbf{trace} on the von Neumann algebra $\A$ is a map $\tau : \A_+\to [0, \infty]$ satisfying:
	\begin{enumerate}[(i)]
		\item $\tau(\lambda T) = \lambda \tau(T)$ for all $T \in \A_+$ and all $\lambda \geq 0.$
		\item $\tau(T + S) = \tau(T) + \tau(S)$ for all $T, S \in \A_+$.
		\item $\tau(U T U^*) = \tau(T)$ for all $T \in \A_+$ and all unitary $U \in \A$.
	\end{enumerate}
\end{definition}

The purpose of this appendix is to prove various claims about trace that were made in that section.
We will prove:
\begin{itemize}
	\item Given properties (i) and (ii), property (iii) is equivalent to the property that for any $O \in \A$, we have $\tau(O O^*) = \tau(O^* O).$
	\item The set $\A_1$ of operators $T \in \A$ for which $\tau(|T|)$ is finite forms an ideal, the \textbf{trace-class operators}.
	\item $\tau$ extends uniquely to a linear functional $\tau : \A_1 \to \comps.$
	\item The extended trace is cyclic on $\A_1$: for $T \in \A_1$ and $S \in \A$, we have $\tau(T S) = \tau(S T).$
\end{itemize}

To prove these facts, we will first work through related facts about the Hilbert-Schmidt operators $\A_2,$ after which these properties will be easy to prove.

\begin{definition}
	An operator $X \in \A$ is said to be \textbf{Hilbert-Schmidt} with respect to $\tau$ if $\tau(X^* X)$ is finite.
	We denote the space of Hilbert-Schmidt operators by $\A_2.$
\end{definition}

\begin{lemma} \label{lem:trace-monotonic}
	For $T, S \in \A_+$ with $T \geq S,$ we have $\tau(T) \geq \tau(S).$
\end{lemma}
\begin{proof}
	We have $\tau(T) = \tau(S) + \tau(T - S) \geq \tau(S).$
\end{proof}

\begin{prop} \label{prop:hilbert-schmidt-ideal}
	$\A_2$ is a two-sided ideal in $\A$.
\end{prop}
\begin{proof}
	Fix $X, Y \in \A_2.$ We have
	\begin{align}
		(X + Y)^* (X + Y)
			& = X^* X + Y^* Y + X^*Y + Y^* X \nonumber \\
			& = 2 X^* X + 2 Y^* Y - (X - Y)^* (X - Y) \nonumber \\
			& \leq 2 X^* X + 2 Y^* Y.
	\end{align}
	So by lemma \ref{lem:trace-monotonic}, we have $\tau((X+Y)^* (X+Y)) \leq 2 \tau(X^* X) + 2 \tau(Y^* Y) < \infty$; it follows that $\A_2$ is closed under addition.
	It is easy to see that $\A_2$ is closed under scalar multiplication.
	So $\A_2$ is a vector space.
	
	To see that $\A_2$ is an ideal, fix $X \in \A_2.$ For any unitary $U \in \A$, we have
	\begin{equation}
		\tau((U X)^* (U X)) = \tau(X^* U^* U X) = \tau(X^* X) < \infty
	\end{equation}
	and
	\begin{equation}
		\tau((X U)^* (X U)) = \tau(U^* X^* X U) = \tau(X^* X) < \infty.
	\end{equation}
	So $\A_2$ is closed under left and right multiplication by unitaries.
	Since every $T \in \A$ can be written as a linear combination of four unitaries (fact \ref{fact:unitaries}), and since we have already shown that $\A_2$ is closed under linear combinations, it follows that $\A_2$ is closed under left and right multiplication by any $T \in \A$.
\end{proof}

\begin{lemma} \label{lem:ideal-closed-under-adjoints}
	$\A_2$ is closed under adjoints.
\end{lemma}
\begin{proof}
	This is a consequence of the more general fact that in any von Neumann algebra, any two-sided ideal is closed under adjoints.
	To see this, let $J$ be a two-sided ideal.
	Then for any $T \in J$ we have the polar decomposition $T = V |T|$ and $T^* = |T| V^* = V^* T V^* \in J.$
\end{proof}

\begin{definition}
	Given an ideal $J$, the symbol $J^2$ denotes the set of all finite linear combinations of products $X Y$ with $X, Y \in J.$
	It is easy to see that this is an ideal.
\end{definition}

\begin{definition}
	We denote by $\A_1$ the ideal $(\A_2)^2$. We call this the \textbf{ideal of trace-class operators}.
	We will soon show that it agrees with our earlier definition.
\end{definition}

\begin{lemma}
	The positive elements of $\A_1$ are exactly the positive operators $T \in \A_+$ with $\tau(T)$ finite.
\end{lemma}
\begin{proof}
	If $T$ is a positive operator with $\tau(T)$ finite, then $|T|^{1/2}$ is in $\A_2,$ so $T = |T|^{1/2} |T|^{1/2}$ is in $\A_1$.
	
	Conversely, if $T$ is a positive operator in $\A_1,$ then $T$ can be written as a sum $X_1^* Y_1 + \dots + X_n^* Y_n$ with $X_j, Y_j \in \A_2.$\footnote{Here we have implicitly used lemma \ref{lem:ideal-closed-under-adjoints}, that $\A_2$ is closed under adjoints.} Each of these satisfies
	\begin{align}
		X_j^* Y_j
			= \frac{1}{4} & \left[ (X_j + Y_j)^* (X_j + Y_j) - (X_j - Y_j)^* (X_j - Y_j)  \right. \nonumber \\
				& \left.+ i (X_j - i Y_j)^* (X_j - i Y_j) - i (X_j + i Y_j)^* (X_j + i Y_j) \right]. \label{eq:four-squares}
	\end{align}
	So $T$ can be written as a linear combination of the form $T = \sum_{j=1}^{m} \alpha_j C_j^* C_j$ with each $C_j$ in $\A_2.$
	Hermiticity of $T$ means that all of the $\alpha_j$ coefficients can be taken to be real.
	So we have $T \leq S$ where $S$ is the sum of all terms in the series for which $\alpha_j$ is positive.
	$\tau(S)$ is finite since each $C_j$ is in $\A_2$, so $\tau(T)$ is finite.
\end{proof}

\begin{prop}
	$\A_1$ is exactly the set of $T \in \A$ such that $\tau(|T|)$ is finite.
\end{prop}

\begin{proof}
	If $T \in \A$ satisfies $\tau(|T|) < \infty,$ then $|T|^{1/2}$ is in $\A_2$.
	For the polar decomposition $T = V |T|,$ we have $V |T|^{1/2}$ in $\A_2$, since $\A_2$ is an ideal and $|T|^{1/2}$ is in $\A_2$. Consequently, $T = (V |T|^{1/2}) |T|^{1/2}$ is in $(\A_2)^2,$ that is $T \in \A_1.$
	
	Conversely, if $T$ is in $\A_1,$ then $|T| = V^* T$ is in $\A_1$ since $\A_1$ is an ideal. So $|T|$ is a positive element of $\A_1,$ and has finite trace by the preceding lemma.
\end{proof}

\begin{remark} \label{rem:four-positives}
	Every operator $T \in \A$ can be written as a preferred linear combination of four positive operators.
	We define the real and imaginary parts of $T$ by
	\begin{align}
		T^R
		& \equiv \frac{T + T^*}{2}, \\
		T^I
		& \equiv \frac{T - T^*}{2 i},
	\end{align}
	then take the positive and negative parts of these Hermitian operators, which are the unique positive operators $T^R_+, T^R_-, T^I_+, T^I_-$ satisfying
	\begin{align}
		T^R
			& = T^R_+ - T^R_-, \\
		T^I
			& = T^I_+ - T^I_-, \\
		0
			& = T^R_+ T^R_-, \\
		0
			& = T^I_+ T^I_-.
	\end{align}
	These can be obtained by applying the continuous functions $f_+(x) = \max\{0, x\}$ and $f_{-}(x) = - \min\{0, x\}$ to the spectra of $T^R$ and $T^I$ using the spectral theorem \ref{thm:C*-spectral}.
	One then has
	\begin{equation}
		T
			= T^R_+ - T^R_- + i T^I_+ - i T^I_-.
	\end{equation}
	It is not hard to see using the spectral theorem that each of these positive operators is dominated by $|T|$, so if $|T|$ is in $\A_1$ then so is each of these positive operators.
\end{remark}

\begin{definition} \label{def:trace-extension}
	We extend the trace $\tau$ to $\A_1$ by
	\begin{equation}
		\tau(T)
			= \tau(T^R_+) - \tau(T^R_-) + i \tau(T^I_+) - i \tau(T^I_-).
	\end{equation}
	By the preceding remark, this sum is well defined, since each term is finite.
\end{definition}

\begin{prop}
	The extension of the trace to $\A_1$ as defined above is linear.
\end{prop}
\begin{proof}
	Fix some $T \in \A_1$, and suppose it can be written as a linear combination
	\begin{equation}
		T = \sum_{j=1}^n \alpha_j S_j
	\end{equation}
	with each $S_j \in \A_1.$
	We would like to show $\tau(T) = \sum_{j=1}^{n} \alpha_j \tau(S_j),$ where $\tau(T)$ and $\tau(S_j)$ are defined as in definition \ref{def:trace-extension}.
	First, we will write each $\alpha_j$ in terms of its real and imaginary parts as $\alpha_j = x_j + i y_j,$ so we have
	\begin{equation}
		T = \sum_{j=1}^n (x_j + i y_j) S_j.
	\end{equation}
	If we write $T$ and each $S_j$ in terms of the preferred linear combinations of positive elements described in remark \ref{rem:four-positives}, we obtain the expression
	\begin{align}
		T^R_+ - T^R_- + i T^I_+ - i T^I_-
			& = \sum_{j=1}^n x_j  (S_{j,+}^R - S_{j,-}^R + i S_{j,+}^I - i S_{j, -}^I) \nonumber \\
			& \quad + i \sum_{j=1}^n  y_j (S_{j,+}^R - S_{j,-}^R + i S_{j,+}^I - i S_{j, -}^I).
	\end{align}
	We can group these elements together so that it schematically looks like
	\begin{equation}
		(\text{sum of pos. terms}) + i (\text{sum of pos. terms})
		= (\text{sum of pos. terms}) + i (\text{sum of pos. terms}).
	\end{equation}
	The first parenthetical is made up of $T_+^R$ together with all of the $x_j S_{j, +}^R$ terms for which $x_j$ is negative, all of the $x_j S_{j, -}^R$ terms for which $x_j$ is positive, all of the $y_j S_{j,+}^I$ terms for which $y_j$ is positive, and all of the $y_j S_{j, -}^I$ terms for which $y_j$ is negative.
	The other pieces of this equation are obtained analogously.
	Taking real and imaginary parts of this equation lets us treat them individually, and then we can use linearity of $\tau$ as applied to each sum of positive terms, which follows from the definition of $\tau$.
	Since all traces involved are finite, we can then regroup the terms to their original sides of the equation to obtain the desired linearity of $\tau$.
\end{proof}

\begin{prop} [Cyclicity of the trace] \label{prop:cyclicity}
	For $T \in \A_1$ and $S \in \A$, we have $\tau(TS) = \tau(ST).$
\end{prop}
\begin{proof}
	For any unitary $U \in \A$ we have, from the definition of $\tau,$ the identity
	\begin{equation}
		\tau(T U) = \tau(U T U U^{*}) = \tau(U T).
	\end{equation}
	The proposition is then proved by writing $S$ as a sum of four unitaries (see fact \ref{fact:unitaries}) and using linearity of the trace on $\A_1.$
\end{proof}

\begin{prop}[Hilbert-Schmidt cyclicity] \label{prop:HS-cyclicity}
	For $X, Y  \in \A_2,$ we have $\tau(XY) = \tau(YX).$
\end{prop}
\begin{proof}
	First we will show $\tau(X X^*) = \tau(X^* X).$
	We take the polar decomposition $X = V |X|.$
	Then we have
	\begin{equation}
		X X^* = V |X|^2 V^*
	\end{equation}
	and
	\begin{equation}
		X^* X = |X|^2.
	\end{equation}
	Since $V^* V$ is the projector onto the closure of the image of $|X|^2,$ we may freely write
	\begin{equation}
		X^* X = V^* V |X|^2.
	\end{equation}
	Now, since $X$ is in $\A_2,$ $|X|^2$ is in $\A_1,$ and since $\A_1$ is an ideal, it follows that $V |X|^2$ is in $\A_1.$
	So using cyclicity of $\tau$ as given in proposition \ref{prop:cyclicity}, we have
	\begin{equation}
		\tau(X^* X) = \tau(V^* V |X|^2) = \tau(V |X|^2 V^*) = \tau(X X^*).
	\end{equation}
	
	Now, for $X, Y \in \A_2$, we recall from equation \eqref{eq:four-squares} that $XY$ can be written as
	\begin{align}
		XY
		= \frac{1}{4}
		&  \left[ (X^* + Y)^*(X^* + Y) - (X^* - Y)^* (X^* - Y) \right. \nonumber \\
		& + \left. i (X^* - i Y)^* (X^* - i Y) - i (X^* + i Y)^* (X^* + i Y) \right].
	\end{align}
	So using linearity of the trace and the result in the preceding paragraph, we see that $\tau(XY)$ is the same as the trace of the operator obtained by switching the two factors in each term on the right-hand side.
	It is easy to check that this is equivalent to the expression obtained by exchanging $X$ and $Y$, so we have $\tau(XY) = \tau(YX).$
\end{proof}

\begin{prop} [Unitary equivalence and positive cyclicity]
	Given conditions (i) and (ii) in definition \ref{def:trace-appendix}, condition (iii) is equivalent to the requirement that $\tau(O O^*) = \tau(O^* O)$ for every $O \in \A.$
\end{prop}
\begin{proof}
	First suppose that we stuck with the original definition of the trace, that $\tau(T) = \tau(U T U^*)$ for all $T \in \A_+$ and all unitary $U \in \A.$
	Then we can apply all of the theorems we have just proved, leading to proposition \ref{prop:HS-cyclicity}, which tells us that if $O$ is Hilbert-Schmidt, then we have $\tau(O^* O) = \tau(O O^*).$
	So the claim holds for all $O$ that are Hilbert-Schmidt.
	But if $O$ is not Hilbert-Schmidt, then by definition $\tau(O^* O)$ is infinite.
	Since $O$ is not Hilbert-Schmidt, $O^*$ is not Hilbert-Schmidt, since by lemma \ref{lem:ideal-closed-under-adjoints} the space of Hilbert-Schmidt operators is closed under adjoints.
	So $\tau(O O^*)$ is also infinite. 
	We thus always have $\tau(O^* O) = \tau(O O^*),$ whether or not $O$ is Hilbert-Schmidt.
	
	Conversely, suppose that we have $\tau(O O^*) = \tau(O^* O)$ for all $O \in \A$. Then for $T \in \A_+$ and unitary $U \in \A,$ we have
	\begin{align}
		\tau(U T U^*)
			& = \tau((U |T|^{1/2}) (U |T|^{1/2})^*) \nonumber \\
			& = \tau((U |T|^{1/2})^* (U |T|^{1/2})) \nonumber \\
			& = \tau(|T|^{1/2} U^* U |T|^{1/2}) \nonumber \\
			& = \tau(|T|) \nonumber \\
			& = \tau(T).
	\end{align}
\end{proof}

\bibliographystyle{JHEP}
\bibliography{bibliography}

\end{document}